\documentclass[12pt]{iopart}
\usepackage{}
\usepackage{cite}
\usepackage{iopams}
\usepackage[ruled,vlined]{algorithm2e}
\usepackage{fancyhdr}
\expandafter\let\csname equation*\endcsname\relax
\expandafter\let\csname endequation*\endcsname\relax
\usepackage{amsmath}
\usepackage{pifont}
\usepackage[dvipsnames]{xcolor}

\begin{document}

\title[Geometry of  Kirkwood-Dirac classical states : A case study based on discrete Fourier transform]{Geometry of Kirkwood-Dirac classical states: A case study based on discrete Fourier transform}

\author{Ying-Hui Yang$^{1}$, Shuang Yao$^{1}$, Shi-Jiao Geng$^{1}$, Xiao-Li Wang$^{1}$ and Pei-Ying Chen$^{1}$}

\address{
$^{1}$ School of Mathematics and Information Science, Henan Polytechnic University, Jiaozuo 454000, China
}

\ead{yangyinghui4149@163.com, yaoshuang0818@163.com}
\vspace{10pt}
\begin{indented}
\item[]
\end{indented}

\begin{abstract}
The characterization of Kirkwood-Dirac (KD) classicality or non-classicality is very important  in quantum information processing. In general, the set of KD classical states with respect to two bases is not a convex polytope[J. Math. Phys. \textbf{65} 072201 (2024)], which makes us interested in finding out in which circumnstances they do form a polytope. In this paper, we focus on the characterization of KD classicality of mixed states for the case where the transition matrix between two bases is a discrete Fourier transform (DFT) matrix in Hilbert space with dimensions $p^2$ and $pq$, respectively, where  $p, q$ are prime. For the two particular cases we investigate, the sets of extremal points are finite, implying that the set of KD classical states we characterize forms a convex polytope. We show that for $p^2$ dimensional system, the set $\rm{KD}_{\mathcal{A},\mathcal{B}}^+$  is a convex hull of the set $\rm {pure}({\rm {KD}_{\mathcal{A},\mathcal{B}}^+})$  based on DFT, where $\rm{KD}_{\mathcal{A},\mathcal{B}}^+$ is the set of KD classical states with respect to  two bases and $\rm {pure}({\rm {KD}_{\mathcal{A},\mathcal{B}}^+})$ is the set of all the rank-one projectors of KD classical pure states  with respect to two bases. In $pq$ dimensional system, we believe that this result also holds. Unfortunately, we do not completely prove it, but some meaningful conclusions are obtained about the characterization of KD classicality.

\end{abstract}

%
\maketitle
%
%

\section{Introduction}

The negativity\cite{negativity.1,negativity.2,negativity.3,Arvidsson.2024,Halpern.2018} or nonreality of quasiprobability distributions\cite{Budiyono.2023}, like entanglement\cite{entanglement}, discord\cite{discord}, coherence\cite{coherence,coherence.2017},  Bell nonlocality\cite{nonlocality} and uncertainty principles\cite{uncertainty}, shows the nonclassical features  of quantum states and plays an important role in quantum information theory and metrology. Kirkwood-Dirac (KD) distribution is a quasiprobability distribution independently discovered by Kirkwood \cite{Kirkwood.1933} and Dirac\cite{Dirac.1945}. KD distributions behave similarly to probability distributions in some aspects, but are allowed to violate some Kolmogorovs axioms\cite{Kolmogorov.1951}. Traditionally, the axiom requiring the distribution function to be nonnegative everywhere on the state space is relaxed, to allow for negative, or even imaginary values. In recent years, KD distribution has become a research hotspot because of its  important applications in  quantum tomography\cite{Lundeen.2012,Bamber.2014,Thekkadath.2016}, weak measurement\cite{Pusey.2014,Dressel.2015,Lupu.2022}, quantum metrology\cite{Lupu.2022,Arvidsson.2020,Jenne.2022}, quantum thermodynamics\cite{Halpern.2018,Levy.2020,Lostaglio.2020}, the relation to nonclassical effects\cite{Spekkens.2008,Lostaglio.2018} and so on. In addition, it has been shown that the KD distribution for states can be extended to a representation of any  finite-dimensional quantum process\cite{Schmid.KD.representations}.

Given a state, its KD distribution is determined by the eigenbases of observables. In a $d$-dimensional Hilbert space, suppose $\boldsymbol{a}=\{|a_i\rangle\}_{i=0}^{d-1}$ and $\boldsymbol{b}=\{|b_j\rangle\}_{j=0}^{d-1}$ are the eigenbases of observables $A$ and $B$, respectively. Let $\mathcal{A}=\{|a_i\rangle\langle a_i|\}_{i=0}^{d-1}$ and $\mathcal{B}=\{|b_j\rangle\langle b_j|\}_{j=0}^{d-1}$. If the KD distribution  of a quantum state with respect to two bases $\boldsymbol{a}, \boldsymbol{b}$ contains negative or nonreal values, then the quantum state is called a  KD nonclassical state; otherwise, the quantum state is called a  KD classical state\cite{ArvidssonJPA.2021,BievrePRL.2021,BievreJMP.2023,Xu.PLA2024,Yang.2023}.

In 2021, Arvidsson-Shuku et al.\cite{ArvidssonJPA.2021} gave a sufficient condition for judging nonclassicality. That is, a state $|\psi\rangle$ is KD nonclassical if $n_{\boldsymbol{a}}(\psi)+n_{\boldsymbol{b}}(\psi)>\lfloor\frac{3d}{2}\rfloor$, where $n_{\boldsymbol{a}}(\psi)$ (respectively $n_{\boldsymbol{b}}(\psi)$) is the number of nonvanishing coefficients of $|\psi\rangle$ on basis $\boldsymbol{a}$ (respectively on basis $\boldsymbol{b}$) and $\lfloor X\rfloor$ is the integer part of $X$.
De Bi\`{e}vre\cite{BievrePRL.2021,BievreJMP.2023} studied completely incompatible observables and their links to the support uncertainty and to KD nonclassicality of a pure state. He  showed that a state $|\psi\rangle$ is KD nonclassical if $n_{\boldsymbol{a}}(\psi)+n_{\boldsymbol{b}}(\psi)>d+1$ and $\langle a_i|b_j\rangle\neq0$ for all $i,j\in\mathbb{Z}_d$.
In 2022, Xu\cite{XuPRA.2022} generalized the concept of completely incompatibility to $s$-order incompatibility. Fiorentino et al.  \cite{Fiorentino.2022} studied the support uncertainty relation of complete sets of mutually unbiased bases. In 2023, for discrete Fourier transform (DFT) matrix as the transition matrix between two bases $\boldsymbol{a}$ and $\boldsymbol{b}$, a pure state $|\psi\rangle$  is real KD classical state if and only if $|\psi\rangle$ satisfies $n_{\boldsymbol{a}}(\psi)n_{\boldsymbol{b}}(\psi)=d$ \cite{Xu.PLA2024,Yang.2023}. In 2023, Langrenez et al.\cite{Langrenez.2023} first analyzed how to characterize KD classical mixed states and showed that for prime $d$ and DFT matrix as the transition matrix, the set $\rm{KD}_{\mathcal{A},\mathcal{B}}^+$ of KD classical states is the convex hull of rank-one projectors associated to the two bases, i.e., $\rm{KD}_{\mathcal{A},\mathcal{B}}^+=$ConvHull$(\mathcal{A}\cup\mathcal{B})$. However, for nonprime $d$, it is still an open question how to characterize KD classical states.

In this paper, we focus on how to characterize KD classical mixed states for nonprime $d$ and the DFT matrix as the transition matrix. We show that for $d=p^{2}$, the set $\rm{KD}_{\mathcal{A},\mathcal{B}}^+$ of KD classical states is the convex hull of rank-one projectors associated to all the KD classical pure states with respect to a pair of bases $\boldsymbol{a}$, $\boldsymbol{b}$, i.e., $\rm{KD}_{\mathcal{A},\mathcal{B}}^+=$ConvHull(pure(KD$^{+}_{\mathcal{A},\mathcal{B}}$)), where pure(KD$^{+}_{\mathcal{A},\mathcal{B}}$) is the set of rank-one projectors associated to all the KD classical pure states with respect to  the pair $\boldsymbol{a}$, $\boldsymbol{b}$ and $p$ is prime. We believe that this result is also true for $d=pq$, where $p$ and $q$ are prime. Although we cannot give a complete proof yet, some relevant results are given. Maybe it is merely a step away from ultimate result.

The rest of this paper is organized as follows. In section 2, some relevant notions are recalled. In section 3.1, we prepare some lemmas to prove the geometry of the set $\rm{KD}_{\mathcal{A},\mathcal{B}}^+$ of KD classical states for $d=p^{2}$. Then we characterize the geometry of KD classical states for $d=p^{2}$ in section 3.2. In section 4, several results about KD classical states for $d=pq$ are given, where $p,q$ are prime and $p\neq q$. Conclusions and discussions are given in section 5.


\section{Preliminaries}
\newtheorem{definition}{Definition}
\newtheorem{lemma}{Lemma}
\newtheorem{theorem}{Theorem}
\newtheorem{corollary}{Corollary}
\newtheorem{proposition}{Proposition}
\newtheorem{conjecture}{Conjecture}
\newtheorem{example}{Example}

\def\QEDclosed{\mbox{\rule[0pt]{1.3ex}{1.3ex}}}
\def\QED{\QEDclosed}
\def\proof{\indent{\em Proof}.}
\def\endproof{\hspace*{\fill}~\QED\par\endtrivlist\unskip}

Consider a $d$-dimensional Hilbert space $\mathcal{H}$. Let $\boldsymbol{a}=\{|a_i\rangle\}_{i=0}^{d-1}$ and $\boldsymbol{b}=\{|b_j\rangle\}_{j=0}^{d-1}$ be two orthonormal bases in $\mathcal{H}$. A matrix $U$ is the transition matrix between the two bases and its entries are $U_{ij}=\langle a_i|b_j\rangle$, where $\ i,j\in\mathbb Z_d$ and $\mathbb{Z}_d=\{0, 1 ,..., d-1\}$. In terms of these two bases $\boldsymbol{a}$ and $\boldsymbol{b}$, the KD distribution of a density matrix $\rho$ can be written as
\begin{equation}
\label{KD-probali}
Q_{ij}(\rho)=\langle b_j|a_i\rangle\langle a_i|\rho|b_j\rangle,\ i,j\in\mathbb Z_d.
\end{equation}
Each value $Q_{ij}(\rho)$ for each choice of $i,j$ and each choice of $\rho$ is also known as a Bargmann
 invariant, specially in the case where $\rho$ is a pure state\cite{Bargmann.2001,Bargmann.2003,Bargmann.2020,Bargmann.2024,Bargmann.arXiv2024}.

KD distribution is a quasi-probability distribution and satisfies
\begin{equation}
\sum_{i,j=0}^{d-1}Q_{ij}(\rho)=1, \sum_{j=0}^{d-1}Q_{ij}(\rho)=\langle a_i|\rho|a_i\rangle, \sum_{i=0}^{d-1}Q_{ij}(\rho)=\langle b_j|\rho|b_j\rangle.\nonumber
\end{equation}
The state $\rho$ is called a KD classical state if its KD distribution is a probability distribution, i.e., $Q_{ij}(\rho)\geq0$ for $\forall i,j\in\mathbb Z_d$.

DFT matrix $U$ is widely used in quantum information processing. In a $d$-dimensional space $\mathcal{H}$, the entries of $U$ are $U_{ij}=\frac{1}{\sqrt{d}}\omega_{d}^{ij}$, where $\omega_{d}=e^{\frac{2\pi\sqrt{-1}}{d}}$ and $i,j\in \mathbb {Z}_{d}$. In this paper, we only consider the case that the transition matrix between the bases $\boldsymbol{a}$ and $\boldsymbol{b}$ is the DFT matrix $U$.

For the DFT matrix $U$ as the transition matrix between two bases $\boldsymbol{a}$ and $\boldsymbol{b}$, the characterization of KD classical pure states has been completed. That is, a state $|\psi\rangle$ is KD classical  if and only if $n_{\boldsymbol{a}}(\psi)n_{\boldsymbol{b}}(\psi)=d$ \cite{Xu.PLA2024,Yang.2023}, where $n_{\boldsymbol{a}}(\psi)$ (respectively $n_{\boldsymbol{b}}(\psi)$) is the number of nonvanishing coefficients of $|\psi\rangle$ on basis $\boldsymbol{a}$ (respectively on basis $\boldsymbol{b}$). While the characterization of KD classical mixed states has not been completed yet. According to Krein-Milman theorem\cite{Hiriart.2004}, the set ${\rm KD}_{\mathcal{A},\mathcal{B}}^+$ of KD classical states with respect to $\boldsymbol{a}$ and  $\boldsymbol{b}$ is the convex hull of its extreme points set $\rm{\rm{ext(\rm{\rm{KD}_{\mathcal{A},\mathcal{B}}^{+}})}}$, i.e.,
\begin{equation}
{\rm{KD}_{\mathcal{A},\mathcal{B}}^+}=\textrm{ConvHull}(\rm{ext(\rm{\rm{KD}_{\mathcal{A},\mathcal{B}}^{+}})}),\nonumber
\end{equation}
where $\mathcal{A}=\{|a_i\rangle\langle a_i|\}_{i=0}^{d-1}$ and $\mathcal{B}=\{|b_j\rangle\langle b_j|\}_{j=0}^{d-1}$. The set $\rm{ext(\rm{\rm{KD}_{\mathcal{A},\mathcal{B}}^{+}})}$ always contains all basis states $\{|a_i\rangle\}_{i=0}^{d-1}$ and $\{|b_j\rangle\}_{j=0}^{d-1}$. Meanwhile,  $\rm{ext(\rm{\rm{KD}_{\mathcal{A},\mathcal{B}}^{+}})}$ may also contain other pure states or mixed states. For  general $\mathcal{A}$ and $\mathcal{B}$, one has ${\rm {pure}}({\rm {KD}}_{\mathcal{A},\mathcal{B}}^+)\subseteq  \rm{ext(\rm{\rm{KD}_{\mathcal{A},\mathcal{B}}^{+}})}$\cite{Langrenez.2023}, where  $\rm {pure}({\rm {KD}_{\mathcal{A},\mathcal{B}}^+})$ is the set of all the projectors of  KD classical pure states with respect to $\boldsymbol{a}$ and $\boldsymbol{b}$, we are interested in investigating when the converse holds. However, it is not easy to obtain a complete description of $\rm{ext(\rm{\rm{KD}_{\mathcal{A},\mathcal{B}}^{+}})}$.  Langrenez et al. showed that for a specific case, that is, prime $d$ and DFT matrix $U$ as the transition matrix, $\rm{KD}_{\mathcal{A},\mathcal{B}}^+=\textrm{ConvHull}(\mathcal{A}\cup\mathcal{B})$  holds\cite{Langrenez.2023}.

For the DFT matrix $U$, KD classical pure states have been completely found out \cite{BievrePRL.2021,Xu.2023}, i.e.,
\begin{eqnarray}
\label{psi-ms}
|\psi_{ms}\rangle=\frac{1}{\sqrt{q}}\sum_{k=0}^{q-1}\omega_q^{sk}|a_{kp+m}\rangle =\frac{1}{\sqrt{p}}\omega_d^{-ms}\sum_{l=0}^{p-1}\omega_p^{-ml}|b_{lq+s}\rangle,\ m\in\mathbb{Z}_p, \ s\in\mathbb{Z}_q, \nonumber\\
\end{eqnarray}
where $d=pq$ and $p,q$ need not be prime. If $p=d$ and $q=1$, then $\{|\psi_{ms}\rangle\}_{m\in\mathbb{Z}_p, s\in\mathbb{Z}_q}=\{|a_i\rangle\}_{i=0}^{d-1}=\boldsymbol{a}$. If $p=1$ and $q=d$, then $\{|\psi_{ms}\rangle\}_{m\in\mathbb{Z}_p, s\in\mathbb{Z}_q}=\{|b_i\rangle\}_{i=0}^{d-1}=\boldsymbol{b}$. The set $\rm {pure}({\rm {KD}_{\mathcal{A},\mathcal{B}}^+})$ can be expressed  as
\begin{eqnarray}
\label{set-pureclassical}
{\rm {pure}}({\rm {KD}}_{\mathcal{A},\mathcal{B}}^+)=\{|\psi_{ms}\rangle\langle\psi_{ms}| \mid m\in\mathbb{Z}_p, s\in\mathbb{Z}_q, pq=d\}.
\end{eqnarray}
Note that in equation (\ref{set-pureclassical}), the divisors $p,q$ take over all the decompositions of $d$. We have the following conjecture for general $d$.

\begin{conjecture}
\label{conjecture}
$\rm{KD}_{\mathcal{A},\mathcal{B}}^+=\rm{ConvHull}(\rm {pure}({\rm {KD}_{\mathcal{A},\mathcal{B}}^+}))$.
\end{conjecture}

This conjecture is true for prime $d$. In this paper, we prove that this conjecture is true for $d=p^2$, and also try our best to prove the conjecture for $d=pq$, where $p$ and $q$ are prime. But the complete proof for $d=pq$ has not been obtained yet.  Nonetheless, some relevant results are obtained. In what follows, $p,q$ are always prime.

If $d=p^2$, then equation (\ref{psi-ms}) can be written as
\begin{eqnarray}
\label{p2-psi}
|\psi_{ms}\rangle=\frac{1}{\sqrt{p}}\sum_{k=0}^{p-1}\omega_p^{sk}|a_{kp+m}\rangle
=\frac{1}{\sqrt{p}}\omega_d^{-ms}\sum_{l=0}^{p-1}\omega_p^{-ml}|b_{lp+s}\rangle,\ m,s\in\mathbb{Z}_p.
\end{eqnarray}
Let $\mathcal{A}=\{|a_i\rangle\langle a_i|\}_{i=0}^{d-1}$, $\mathcal{B}=\{|b_j\rangle\langle b_j|\}_{j=0}^{d-1}$ and $\mathcal{C}=\{|\psi_{ms}\rangle \langle\psi_{ms}|\}_{m,s\in\mathbb{Z}_p}$. In fact, the three sets $\mathcal{A},\mathcal{B},\mathcal{C}$ correspond to three types of factorizations of $d$, i.e., $d=d\times 1$, $d=1\times d$ and $d=p\times p$, respectively. Thus, $\rm {pure}({\rm {KD}_{\mathcal{A},\mathcal{B}}^+})=\mathcal{A}\cup\mathcal{B}\cup\mathcal{C}$.

If $d=pq$ and $p\neq q$, then  besides the basis states $\boldsymbol{a}$ and $\boldsymbol{b}$, there are another two types of KD classical states by equation (\ref{psi-ms}), i.e.,
\numparts
\begin{eqnarray}
\label{pq-psi}
|\psi_{ms}\rangle&=\frac{1}{\sqrt{q}}\sum_{k=0}^{q-1}\omega_q^{sk}|a_{kp+m}\rangle \nonumber\\
&=\frac{1}{\sqrt{p}}\omega_d^{-ms}\sum_{l=0}^{p-1}\omega_p^{-ml}|b_{lq+s}\rangle,\ m\in\mathbb{Z}_p, \ s\in\mathbb{Z}_q,  \\
\label{pq-varphi}
|\varphi_{m's'}\rangle&=\frac{1}{\sqrt{p}}\sum_{k=0}^{p-1}\omega_p^{s'k}|a_{kq+m'}\rangle  \nonumber\\
&=\frac{1}{\sqrt{q}}\omega_d^{-m's'}\sum_{l=0}^{q-1}\omega_q^{-m'l}|b_{lp+s'}\rangle, \ m'\in\mathbb{Z}_q, \ s'\in\mathbb{Z}_p.
\end{eqnarray}
\endnumparts
Let $\mathcal{A}$ and $\mathcal{B}$ still be projector sets of basis states $\boldsymbol{a}$ and $\boldsymbol{b}$, respectively. Let $\mathcal{C}=\{|\psi_{ms}\rangle \langle\psi_{ms}|\}_{m\in\mathbb{Z}_p,s\in\mathbb{Z}_q}$ and $\mathcal{D}=\{|\varphi_{m's'}\rangle\langle\varphi_{m's'}|\}_{m'\in\mathbb{Z}_q, s'\in\mathbb{Z}_p}$. Here, the sets $\mathcal{C}, \mathcal{D}$ correspond to $d=q\times p$ and $d=p\times q$, respectively. Obviously, $\rm {pure}({\rm {KD}_{\mathcal{A},\mathcal{B}}^+})=\mathcal{A}\cup\mathcal{B}\cup\mathcal{C}\cup\mathcal{D}$.

In what follows, we try to prove the conjecture for $d=p^{2}$ and $d=pq$, respectively.

\section{Geometry of Kirkwood-Dirac classical states for $d=p^2$}

\subsection{Characterizing $\rm{span}_{\mathbb{R}}(\mathcal{A}\cup\mathcal{B}\cup\mathcal{C})$ and $\rm{KD}_{\mathcal{A},\mathcal{B}}^r$}
In this section, we will adopt a research framework similar to that in  reference \cite{Langrenez.2023}. We first study the vector space $\rm{span}_{\mathbb{R}}(\mathcal{A}\cup\mathcal{B}\cup\mathcal{C})$ and explore the relation among $\rm{KD}_{\mathcal{A},\mathcal{B}}^+$, $\rm{span}_{\mathbb{R}}(\mathcal{A}\cup\mathcal{B}\cup\mathcal{C})$ and ${\rm{ConvHull}}(\mathcal{A}\cup\mathcal{B}\cup\mathcal{C})$. For simplicity, the projector of a pure state $|\varphi\rangle$ is denoted by $\varphi:=|\varphi\rangle\langle \varphi|$ in the whole paper.

\begin{lemma}
\label{lemma-1}
If $d=p^2$, then $\dim {\rm span}_{\mathbb{R}}(\mathcal{A}\cup \mathcal{B}\cup \mathcal{C})=3p^2-2p$ and $\rm{KD}_{\mathcal{A},\mathcal{B}}^+\cap span_{\mathbb{R}}(\mathcal{A}\cup \mathcal{B}\cup \mathcal{C})=\rm{ConvHull}(\mathcal{A}\cup \mathcal{B}\cup \mathcal{C})$.
\end{lemma}
\begin{proof}
Let us first show the dimension of $\textrm{span}_{\mathbb{R}}(\mathcal{A}\cup \mathcal{B}\cup \mathcal{C})$. Define a linear map $\Gamma:\mathbb{R}^{3p^2}\rightarrow \rm{span}_{\mathbb{R}}(\mathcal{A}\cup \mathcal{B}\cup \mathcal{C})$, for each $((\lambda_i),(\mu_j),(\gamma_{ms}))\in\mathbb{R}^{3p^2}$,
\begin{equation}
\Gamma((\lambda_i),(\mu_j),(\gamma_{ms}))=\sum_{i=0}^{p^2-1}\lambda_ia_i+\sum_{j=0}^{p^2-1}\mu_jb_j+\sum_{m,s=0}^{p-1}\gamma_{ms}\psi_{ms},\nonumber
\end{equation}
where $((\lambda_i),(\mu_j),(\gamma_{ms}))=(\lambda_0,\ldots,\lambda_{p^2-1},\mu_0,\ldots,\mu_{p^2-1},\gamma_{00},\ldots,\gamma_{(p-1)(p-1)})$. Note that here we use the notation $a_i=|a_i\rangle\langle a_i|$.
From the rank theorem, we have $\dim\rm{span}_{\mathbb{R}}(\mathcal{A}\cup \mathcal{B}\cup \mathcal{C})$ $= 3p^2-\dim(\rm{Ker} (\Gamma))$.

Suppose
\begin{equation}
\Upsilon=\sum_{i=0}^{p^2-1}\lambda_ia_i+\sum_{j=0}^{p^2-1}\mu_jb_j+\sum_{m,s=0}^{p-1}\gamma_{ms}\psi_{ms}=0, \nonumber
\end{equation}
where $\lambda_i,\mu_j,\gamma_{ms}\in \mathbb{R}$. We will calculate $\dim(\rm{Ker} (\Gamma))$ below.
Since $\Upsilon=0$, one has
\begin{equation}
\langle a_i|\Upsilon|b_j\rangle=\langle a_i|b_j\rangle(\lambda_i+\mu_j+\gamma_{ms}) = 0,\ i,j\in\mathbb Z_{p^2},\nonumber
\end{equation}
where $i \equiv m \mod p$, $j \equiv s \mod p$, i.e., $i = kp+m$, $j = lp+s$, where $k, l\in\mathbb Z_p$.
Since $\langle a_i|b_j\rangle =\frac{1}{\sqrt{d}}\omega_{d}^{ij} \neq 0$ , it follows $\lambda_i+\mu_j+\gamma_{ms} = 0$, i.e.,
\begin{equation}
\label{lam-mu-ga}
\lambda_{kp+m}+\mu_{lp+s}+\gamma_{ms} = 0, \ \textrm{for} \ \forall\ k,l,m,s\in\mathbb Z_p.
\end{equation}
It follows $\lambda_{m}+\mu_{lp+s}+\gamma_{ms}=\lambda_{kp+m}+\mu_{lp+s}+\gamma_{ms}$. Thus $\lambda_m = \lambda_{kp+m}$, for $k\in\mathbb Z_p$.
Similarly, $\mu_s = \mu_{lp+s}$, for $l\in\mathbb Z_p$.
It implies that equation (\ref{lam-mu-ga}) can be written as $\lambda_m+\mu_s+\gamma_{ms} = 0$, i.e.,
\begin{equation}
\label{ms_}
\gamma_{ms} = -(\lambda_m+\mu_s).
\end{equation}
From equation (\ref{ms_}), one can obtain $\dim(\rm{Ker} (\Gamma))=$ $2p$, and then $\dim\rm{span}_{\mathbb{R}}(\mathcal{A}\cup \mathcal{B}\cup \mathcal{C}) =$ $3p^2 - 2p$.

Next, we prove the second statement. By equation (\ref{KD-probali}), it is easy to verify $\rm{ConvHull}(\mathcal{A}\cup \mathcal{B}\cup \mathcal{C})\subseteq \rm{KD}_{\mathcal{A},\mathcal{B}}^+$. Together with the fact that $\rm{ConvHull}(\mathcal{A}\cup \mathcal{B}\cup \mathcal{C})\subseteq \rm{span}_{\mathbb{R}}(\mathcal{A}\cup \mathcal{B}\cup \mathcal{C})$, this shows $\rm{ConvHull}(\mathcal{A}\cup \mathcal{B}\cup \mathcal{C})\subseteq \rm{KD}_{\mathcal{A},\mathcal{B}}^+\cap \rm{span}_{\mathbb{R}}(\mathcal{A}\cup \mathcal{B}\cup \mathcal{C})$. Now we only need to prove $\rm{KD}_{\mathcal{A},\mathcal{B}}^+\cap \rm{span}_{\mathbb{R}}(\mathcal{A}\cup \mathcal{B}\cup \mathcal{C})\subseteq \rm{ConvHull}(\mathcal{A}\cup \mathcal{B}\cup \mathcal{C})$.

Suppose $\rho\in\rm{KD}_{\mathcal{A},\mathcal{B}}^+\cap \rm{span}_{\mathbb{R}}(\mathcal{A}\cup \mathcal{B}\cup \mathcal{C})$. There exists a vector $((\lambda_i),(\mu_j),(\gamma_{ms}))\in\mathbb{R}^{3p^2}$ so that
\begin{equation}
\label{eq2}
\rho=\sum_{i=0}^{p^2-1}\lambda_ia_i+\sum_{j=0}^{p^2-1}\mu_jb_j+\sum_{m,s=0}^{p-1}\gamma_{ms}\psi_{ms}.
\end{equation}
Since $\rho\in\rm{KD}_{\mathcal{A},\mathcal{B}}^+$, it follows
\begin{equation}
Q_{ij}(\rho)=|\langle a_i|b_j\rangle|^2(\lambda_i+\mu_j+\gamma_{ms}) \geqslant 0,\ \forall\ i,j\in\mathbb Z_{p^2}, \nonumber
\end{equation}
where $i \equiv m \mod p$, $j \equiv s \mod p$, i.e., $i=kp + m, j=lp + s$, where $k,l\in \mathbb Z_p$. Since $|\langle a_i|b_j\rangle|^2\geq0$, it follows  $\lambda_i+\mu_j+\gamma_{ms} \geq 0$, i.e.,
\begin{equation}
\label{eq-gamma}
\lambda_{kp + m}+\mu_{lp + s}+\gamma_{ms} \geq 0,\ \forall\ m,s,k,l\in \mathbb Z_p.
\end{equation}
Equation (\ref{eq2}) can be written as
\begin{equation}
\label{eq3}
\rho=\sum_{m,k=0}^{p-1}\lambda_{kp + m}a_{kp + m}+\sum_{s,l=0}^{p-1}\mu_{lp + s}b_{lp + s}+\sum_{m,s=0}^{p-1}\gamma_{ms}\psi_{ms}.
\end{equation}

In order to see that all the coefficients of $a_{kp + m}$ and $b_{lp + s}$ are non-negative, without loss of generality, suppose that ${\rm min}_{k\in \mathbb Z_p}\{\lambda_{kp+m}\}=\lambda_m$ for $m\in \mathbb Z_p$, ${\rm min}_{l\in \mathbb Z_p}\{\mu_{lp+s}\}=\mu_s$  for $s\in \mathbb Z_p$. Then equation (\ref{eq3}) can be written as
\begin{eqnarray}
\label{eq4}
\rho &=&\sum_{m,k=0}^{p-1}(\lambda_{kp+m}-\lambda_m)a_{kp+m}+\sum_{s,l=0}^{p-1}(\mu_{lp+s}-\mu_s)b_{lp+s} \nonumber\\
&\;&+\sum_{m,k=0}^{p-1}\lambda_ma_{kp+m} +\sum_{s,l=0}^{p-1}\mu_sb_{lp+s}+\sum_{m,s=0}^{p-1}\gamma_{ms}\psi_{ms}.
\end{eqnarray}
Notice that $\lambda_{kp+m}-\lambda_m\geq0$ and $\mu_{lp+s}-\mu_s\geq0$.

Now let us consider the relation among the projectors $a_{kp+m}$, $b_{lp+s}$ and $\psi_{ms}$.
By equation (\ref{p2-psi}), we have
\begin{eqnarray}
\psi_{ms}&=&\frac{1}{p}(\sum_{k=0}^{p-1}a_{kp+m}+\sum_{k_1\neq k_2}\omega_p^{s(k_1-k_2)} |a_{k_{1}p+m}\rangle\langle a_{k_{2}p+m}|) \nonumber\\
&=&\frac{1}{p}(\sum_{l=0}^{p-1}b_{lp+s}+\sum_{l_1\neq l_2}\omega_p^{-m(l_1-l_2)}|b_{l_{1}p+s}\rangle\langle b_{l_{2}p+s}|).\nonumber
\end{eqnarray}
Hence,
\numparts
\begin{equation}
\label{a-{ms}}
\forall m\in\mathbb Z_p, \sum_{s=0}^{p-1}\psi_{ms}=\sum_{k=0}^{p-1}a_{kp+m},
\end{equation}
\begin{equation}
\label{b-{ms}}
\forall s\in\mathbb Z_p, \sum_{m=0}^{p-1}\psi_{ms}=\sum_{l=0}^{p-1}b_{lp+s}.
\end{equation}
\endnumparts
Substituting equations (\ref{a-{ms}}) and (\ref{b-{ms}}) into equation (\ref{eq4}), equation (\ref{eq4}) can be written as
\begin{equation}
\label{eq6}
\rho=\sum_{m,k=0}^{p-1}(\lambda_{kp+m}-\lambda_m)a_{kp+m}+\sum_{s,l=0}^{p-1}(\mu_{lp+s}-\mu_s)b_{lp+s}
+\sum_{m,s=0}^{p-1}(\lambda_m+\mu_s+\gamma_{ms})\psi_{ms}.
\end{equation}
Note that $\lambda_m+\mu_s+\gamma_{ms}\geq0$ from equation (\ref{eq-gamma}). Then all coefficients of $\rho$ in equation (\ref{eq6}) are non-negative. Hence $\rho\in {\rm span}_{\mathbb{R^{+}}}(\mathcal{A}\cup \mathcal{B}\cup \mathcal{C})$. Thus $\rho\in {\rm ConvHull}(\mathcal{A}\cup \mathcal{B}\cup \mathcal{C})$ since ${\rm Tr}\rho=1$. This completes the proof.
\end{proof}

A self-adjoint operator $F$ is called a KD real operator if $Q_{ij}(F)$ are real for $\forall i,j\in\mathbb Z_d$. We denote the set of KD real operators by $\rm{KD}_{\mathcal{A},\mathcal{B}}^r$. It is easy to verify that $\rm{KD}_{\mathcal{A},\mathcal{B}}^r$ is a linear space. Obviously, $\rm{KD}_{\mathcal{A},\mathcal{B}}^+\subseteq \rm{KD}_{\mathcal{A},\mathcal{B}}^r$ since a KD classical state is KD real.
For all $\rho\in  {\rm span}_{\mathbb{R}}(\mathcal{A}\cup \mathcal{B}\cup\mathcal{C})$,  $Q_{ij}(\rho)$ are real for all $\ i,j\in\mathbb Z_d$, hence ${\rm span}_{\mathbb{R}}(\mathcal{A}\cup \mathcal{B}\cup\mathcal{C})\subseteq
\rm{KD}_{\mathcal{A},\mathcal{B}}^r$. In reference \cite{Langrenez.2023}, the authors showed that $\rm{KD}_{\mathcal{A},\mathcal{B}}^r = {\rm span}_{\mathbb{R}}(\mathcal{A}\cup \mathcal{B})$ if and only if $\rm{KD}_{\mathcal{A},\mathcal{B}}^+= {\rm ConvHull}(\mathcal{A}\cup \mathcal{B})$ for prime $d$. Here, applying lemma \ref{lemma-1}, we obtain a similar conclusion for $d=p^2$.
\begin{lemma}
\label{lemma-3}
$\rm{KD}_{\mathcal{A},\mathcal{B}}^r = span_{\mathbb{R}}({\mathcal{A}\cup \mathcal{B}\cup \mathcal{C}})$ if and only if $\rm{KD}_{\mathcal{A},\mathcal{B}}^+= ConvHull(\mathcal{A}\cup \mathcal{B}\cup \mathcal{C})$ for $d=p^2$.
\end{lemma}
\begin{proof}
We first prove the necessity. Since ${\rm ConvHull}(\mathcal{A}\cup \mathcal{B}\cup \mathcal{C})\subseteq \rm{KD}_{\mathcal{A},\mathcal{B}}^+$ is obviously true, we only need to show $\rm{KD}_{\mathcal{A},\mathcal{B}}^+\subseteq {\rm ConvHull}(\mathcal{A}\cup \mathcal{B}\cup \mathcal{C})$. Let $\rho\in \rm{KD}_{\mathcal{A},\mathcal{B}}^+$, then we get that $\rho\in \rm{KD}_{\mathcal{A},\mathcal{B}}^r$. Thus $\rho\in {\rm span}_{\mathbb{R}}(\mathcal{A}\cup \mathcal{B}\cup \mathcal{C})$ since $\rm{KD}_{\mathcal{A},\mathcal{B}}^r = {\rm span}_{\mathbb{R}}(\mathcal{A}\cup \mathcal{B}\cup \mathcal{C})$. Applying lemma \ref{lemma-1}, it follows $\rho\in {\rm ConvHull}(\mathcal{A}\cup \mathcal{B}\cup \mathcal{C})$, which means $\rm{KD}_{\mathcal{A},\mathcal{B}}^+\subseteq {\rm ConvHull}(\mathcal{A}\cup \mathcal{B}\cup \mathcal{C})$. The desired result is obtained.

The proof of sufficiency is almost the same as that in reference \cite{Langrenez.2023}. For the convenience of readers and the completeness of the article, it is given in \ref{Proof of Lemma 3}.
\end{proof}

For prime $d$, a self-adjoint operator $F \in \rm{KD}_{\mathcal{A},\mathcal{B}}^r$ if and only if the entries of $F$ satisfy $ F_{i(i+k)}=F_{(i-k)i}$ for all $i,k \in\mathbb Z_d$ \cite{Langrenez.2023}. We improve the proof method of this result in reference \cite{Langrenez.2023} and get the following lemma. The following lemma shows that this result is always true whether $d$ is prime or not.  The lemma is a powerful tool to prove  conjecture \ref{conjecture}. It can help us to better understand $\rm{KD}_{\mathcal{A},\mathcal{B}}^r$.
\begin{lemma}
\label{lemma-4}
For any dimension $d$, a self-adjoint operator $F \in \rm{KD}_{\mathcal{A},\mathcal{B}}^r$ if and only if
\begin{equation}
\label{eq-F}
F_{i(i+k)}=F_{(i-k)i}, \forall \ i,k \in\mathbb Z_d,
\end{equation}
where $F_{ik}=\langle a_i|F|a_k\rangle$.
\end{lemma}

The proof of lemma \ref{lemma-4} is given in  \ref{The proof of lemma 4}.

If a self-adjoint operator $F \in \rm{KD}_{\mathcal{A},\mathcal{B}}^r$, then the diagonal entries $F_{ii}$ are always real and
\begin{equation}
\label{F-k-(d-k)}
F_{i(i+k)}=\overline {F_{(i+k)i}}=\overline {F_{(i-(d-k))i}} = \overline {F_{i(i+(d-k))}}
\end{equation}
since equation (\ref{eq-F}) holds. Equation (\ref{F-k-(d-k)}) implies that $F_{i(i+k)}$ and $F_{i(i+(d-k))}$ are conjugate. Suppose that all the diagonal entries $F_{ii}$ are as one category. By lemma \ref{lemma-4}, off-diagonal entries in $F$ can be classified into some distinct categories, where the same or conjugate entries are in a common category. Now we show how to classify the off-diagonal entries in $F$ by equation (\ref{eq-F}).

Equation (\ref{eq-F}) implies that for $\forall\ i\in\mathbb Z_{d}$,
\begin{equation}
\label{F-k}
F_{(i-k)i}=F_{i(i+k)}=\cdots=F_{(i+lk)(i+(l+1)k)}=\cdots=F_{(i+(d-1)k)(i+dk)}, l\in \mathbb Z_{d}.
\end{equation}
Notice that $i-k\equiv i+(d-1)k \mod d$ and $i\equiv i+dk \mod d$. Suppose gcd$(k,d)=t$, where gcd$(k,d)$ is the greatest common divisor of $k$ and $d$.

If $t=1$, then $i+lk\neq i+l'k \mod d$ for $l,l'\in \mathbb Z_{d}$ and $l\neq l'$. Otherwise, $l\equiv l' \mod d$. It implies that $d$ entries are equal in equation (\ref{F-k}). For the convenience of description, in the whole paper, we use ``length" to denote the number of equal elements of matrix $F$ in an equation. It follows the length of equation (\ref{F-k}) is $d$.

If $t>1$, then $i+lk\equiv i+l'k \mod d$ if and only if $l\equiv l' \mod \frac{d}{t}$. Thus, equation (\ref{F-k}) can be written as
\begin{equation}
\label{F-kt}
F_{(i-k)i}=F_{i(i+k)}=F_{(i+k)(i+2k)}=\cdots=F_{(i+(\frac{d}{t}-1)\times k)(i+\frac{d}{t}\times k)}.
\end{equation}
Since $i-k\equiv i+(\frac{d}{t}-1)\times k \mod d$ and $i\equiv i+\frac{d}{t}\times k \mod d$ in equation (\ref{F-kt}), the length of equation (\ref{F-kt}) is $\frac{d}{t}$.

In order to understand how to classify the off-diagonal entries in $F$ by equation (\ref{eq-F}) more intuitively,  and to understand more easily the proof of theorem \ref{theorem}, the following three examples are given.

\label{Three examples}
\begin{example}
\label{d=5}
The entries of a self-adjoint operator $F \in \rm{KD}_{\mathcal{A},\mathcal{B}}^r$ for $d=5$ can be classified into three categories.
\end{example}

Off-diagonal entries in $F$ can be classified into two categories.
From equation (\ref{F-k}) we have
\begin{equation}
\label{F-k-d=5}
F_{(i-k)i}=F_{i(i+k)}=F_{(i+k)(i+2k)}=\cdots=F_{(i+4k)(i+5k)},\forall\  i,k\in\mathbb Z_5.
\end{equation}
Since gcd$(k,5)=1$ for all $k\in\mathbb Z^*_5$, the length of equation (\ref{F-k-d=5}) is 5. For $k=1$ and $k=4$, we have $F_{i(i+1)}= \overline {F_{i(i+4)}}$ by equation (\ref{F-k-(d-k)}). Thus $F_{i(i+1)}$ and $F_{i(i+4)}$ are in a common category which has $10$ entries and is marked with green in matrix $F$. Similarly, for $k=2$ and $k=3$, $F_{i(i+2)}$ and $F_{i(i+3)}$ marked with red are in a common category. Together with the category of diagonal entries, there are three categories. See the following matrix $F$.
\begin{equation}
F=
\begin{pmatrix}
F_{00}& {\color{green}F_{01}}&{\color{red}F_{02}}&{\color{red}F_{03}}&{\color{green}F_{04}}\\
{\color{green}F_{10}}&F_{11}&{\color{green}F_{12}}&{\color{red}F_{13}}&{\color{red}F_{14}}\\
{\color{red}F_{20}}&{\color{green}F_{21}}&F_{22}&{\color{green}F_{23}}&{\color{red}F_{24}}\\
{\color{red}F_{30}}&{\color{red}F_{31}}&{\color{green}F_{32}}&F_{33}&{\color{green}F_{34}}\\
{\color{green}F_{40}}&{\color{red}F_{41}}&{\color{red}F_{42}}&{\color{green}F_{43}}&F_{44} \nonumber
\end{pmatrix}
\end{equation}

\begin{example} The entries of a self-adjoint operator $F \in \rm{KD}_{\mathcal{A},\mathcal{B}}^r$ for $d=9$ can be classified into seven categories.
\end{example}

Off-diagonal entries in $F$ are classified into six categories.
Consider two cases gcd$(k,9)=1$ and gcd$(k,9)=3$.

Case 1. gcd$(k,9)=1$, i.e., $k=1,2,4,5,7$ or $8$. It means that the length of equation (\ref{F-k}) is 9. By equation (\ref{F-k-(d-k)}), we have $F_{i(i+1)}= \overline {F_{i(i+8)}}$.
Thus $F_{i(i+1)}$ and $F_{i(i+8)}$ for all $i\in\mathbb Z_9$ are in a common category which has $18$ entries and is marked with orange color in matrix $F$. Similarly, $F_{i(i+2)}$ and $F_{i(i+7)}$, $F_{i(i+4)}$ and $F_{i(i+5)}$ are in a common category and are marked with red and blue color, respectively. There are three categories for this case.

Case 2. gcd$(k,9)=3$, i.e., $k=3$ or $6$. From equation (\ref{F-kt}) we have
\begin{equation}
\label{F-kt-d=9}
F_{(i-k)i}=F_{i(i+k)}=F_{(i+k)(i+2k)}=F_{(i+2k)(i+3k)}.
\end{equation}
The length of equation (\ref{F-kt-d=9}) is $3$. For $k=3$, equation (\ref{F-kt-d=9}) implies  $F_{03}=F_{36}=F_{60}$ for $i\equiv0 \mod 3$, $F_{14}=F_{47}=F_{71}$ for $i\equiv1 \mod 3$ and $F_{25}=F_{58}=F_{82}$ for $i\equiv2 \mod 3$. By equation (\ref{F-k-(d-k)}), $F_{i(i+3)}=\overline {F_{i(i+6)}}$ for all $i\in\mathbb Z_9$. Thus
\begin{eqnarray}
F_{03}=F_{36}=F_{60}=\overline {F_{30}}= \overline {F_{63}}=\overline {F_{06}}, \nonumber\\
F_{14}=F_{47}=F_{71}=\overline {F_{41}}= \overline {F_{74}}=\overline {F_{17}}, \nonumber\\
F_{25}=F_{58}=F_{82}=\overline {F_{52}}= \overline {F_{85}}=\overline {F_{28}}. \nonumber
\end{eqnarray}
There are three categories, marked with green, cyan and purple, respectively.
Together with the category of diagonal entries, there are seven categories. See the following matrix $F$.

\begin{equation}
F=
\begin{pmatrix}
F_{00}& \textcolor{orange}{F_{01}}&{\color{red}F_{02}}&{\color{green}F_{03}}&{\color{blue}F_{04}}&{\color{blue}F_{05}}&{\color{green}F_{06}}&{\color{red}F_{07}}&\textcolor{orange}{F_{08}}\\
\textcolor{orange}{F_{10}}&F_{11}&\textcolor{orange}{F_{12}}&{\color{red}F_{13}}&\textcolor{cyan}{F_{14}}&{\color{blue}F_{15}}&{\color{blue}F_{16}}&\textcolor{cyan}{F_{17}}&{\color{red}F_{18}}\\
{\color{red}F_{20}}&\textcolor{orange}{F_{21}}&F_{22}&\textcolor{orange}{F_{23}}&{\color{red}F_{24}}&{\color{purple}F_{25}}&{\color{blue}F_{26}}&{\color{blue}F_{27}}&{\color{purple}F_{28}}\\
{\color{green}F_{30}}&{\color{red}F_{31}}&\textcolor{orange}{F_{32}}&F_{33}&\textcolor{orange}{F_{34}}&{\color{red}F_{35}}&{\color{green}F_{36}}&{\color{blue}F_{37}}&{\color{blue}F_{38}}\\
{\color{blue}F_{40}}&\textcolor{cyan}{F_{41}}&{\color{red}F_{42}}&\textcolor{orange}{F_{43}}&F_{44}&\textcolor{orange}{F_{45}}&{\color{red}F_{46}}&\textcolor{cyan}{F_{47}}&{\color{blue}F_{48}}\\
{\color{blue}F_{50}}&{\color{blue}F_{51}}&{\color{purple}F_{52}}&{\color{red}F_{53}}&\textcolor{orange}{F_{54}}&F_{55}&\textcolor{orange}{F_{56}}&{\color{red}F_{57}}&{\color{purple}F_{58}}\\
{\color{green}F_{60}}&{\color{blue}F_{61}}&{\color{blue}F_{62}}&{\color{green}F_{63}}&{\color{red}F_{64}}&\textcolor{orange}{F_{65}}&F_{66}&\textcolor{orange}{F_{67}}&{\color{red}F_{68}}\\
{\color{red}F_{70}}&\textcolor{cyan}{F_{71}}&{\color{blue}F_{72}}&{\color{blue}F_{73}}&\textcolor{cyan}{F_{74}}&{\color{red}F_{75}}&\textcolor{orange}{F_{76}}&F_{77}&\textcolor{orange}{F_{78}}\\
\textcolor{orange}{F_{80}}&{\color{red}F_{81}}&{\color{purple}F_{82}}&{\color{blue}F_{83}}&{\color{blue}F_{84}}&{\color{purple}F_{85}}&{\color{red}F_{86}}&\textcolor{orange}{F_{87}}&F_{88}\nonumber
\end{pmatrix}
\end{equation}

\begin{example} The entries of a self-adjoint operator $F \in \rm{KD}_{\mathcal{A},\mathcal{B}}^r$ for $d=6$ can be classified into seven categories.
\end{example}

Off-diagonal entries in $F$ are classified into six categories. Consider three cases gcd$(k,6)=1$, gcd$(k,6)=2$ and gcd$(k,6)=3$.

Case 1. gcd$(k,d)=1$, i.e., $k=1$ or $5$.  It implies that the length of equation (\ref{F-k}) is 6.
By equation (\ref{F-k-(d-k)}), we have $F_{i(i+1)}= \overline {F_{i(i+5)}}$. Thus $F_{i(i+1)}$ and $F_{i(i+5)}$ for all $i\in\mathbb Z_6$ are in a common category which has $12$ entries, marked with cyan color in matrix $F$. There is one category for this case.

Case 2. gcd$(k,6)=2$, i.e., $k=2$ or $4$. From equation (\ref{F-kt}) we have
\begin{equation}
\label{F-kt-d=6}
F_{(i-k)i}=F_{i(i+k)}=F_{(i+k)(i+2k)}=F_{(i+2k)(i+3k)}.
\end{equation}
The length of equation (\ref{F-kt-d=6}) is $3$. For $k=2$, we have $F_{02}=F_{24}=F_{40}$ for $i\equiv0 \mod 2$ and $F_{13}=F_{35}=F_{51}$ for $i\equiv1 \mod 2$ by equation (\ref{F-kt-d=6}). By equation (\ref{F-k-(d-k)}), $F_{i(i+2)}=\overline {F_{i(i+4)}}$ for all $i\in\mathbb Z_6$. Thus
\begin{eqnarray}
F_{02}=F_{24}=F_{40}=\overline {F_{20}}= \overline {F_{42}}=\overline {F_{04}}, \nonumber\\
F_{13}=F_{35}=F_{51}= \overline {F_{31}}=\overline {F_{53}}=\overline {F_{51}}. \nonumber
\end{eqnarray}
There are two categories, marked with red and blue, respectively.

Case 3. gcd$(k,6)=3$, i.e., $k=3$. From equation (\ref{F-kt}) we have
\begin{equation}
\label{gcd=3-d=6}
F_{i(i+3)}=F_{(i+3)i}.
\end{equation}
The length of equation (\ref{gcd=3-d=6}) is $2$.  Then $F_{i(i+3)}=\overline {F_{i(i+3)}}$ for all $i\in\mathbb Z_6$ by equation (\ref{F-k-(d-k)}). It indicates that $F_{i(i+3)}$ is real. Equation (\ref{gcd=3-d=6}) implies $F_{03}=F_{30}$ for $i\equiv0 \mod 3$, $F_{14}=F_{41}$ for $i\equiv1 \mod 3$ and $F_{25}=F_{52}$ for $i\equiv2 \mod 3$. There are three categories, marked with green, orange and purple, respectively.
Together with the category of diagonal entries, there are seven categories. See the following matrix $F$.

\begin{equation}
F=
\begin{pmatrix}
F_{00}&\textcolor{cyan}{F_{01}}&{\color{red}F_{02}}&{\color{green}F_{03}}&{\color{red}F_{04}}&\textcolor{cyan}{F_{05}}\\
\textcolor{cyan}{F_{10}}&F_{11}&\textcolor{cyan}{F_{12}}&{\color{blue}F_{13}}&\textcolor{orange}{F_{14}}&{\color{blue}F_{15}}\\
{\color{red}F_{20}}&\textcolor{cyan}{F_{21}}&F_{22}&\textcolor{cyan}{F_{23}}&{\color{red}F_{24}}&{\color{purple}F_{25}}\\
{\color{green}F_{30}}&{\color{blue}F_{31}}&\textcolor{cyan}{F_{32}}&F_{33}&\textcolor{cyan}{F_{34}}&{\color{blue}F_{35}}\\
{\color{red}F_{40}}&\textcolor{orange}{F_{41}}&{\color{red}F_{42}}&\textcolor{cyan}{F_{43}}&F_{44}&\textcolor{cyan}{F_{45}}\\
\textcolor{cyan}{F_{50}}&{\color{blue}F_{51}}&{\color{purple}F_{52}}&{\color{blue}F_{53}}&\textcolor{cyan}{F_{54}}&F_{55}\nonumber
\end{pmatrix}
\end{equation}

\subsection{Characterizing KD classical states}

The set of KD classical states is $\rm{KD}_{\mathcal{A},\mathcal{B}}^+={\rm ConvHull}(\mathcal{A}\cup\mathcal{B})$ when $d$ is prime and $U$ is the DFT matrix\cite{Langrenez.2023}. Now we show that  conjecture \ref{conjecture} is correct for $d=p^2$.
\begin{theorem}
\label{theorem}
Suppose that $U$ is the $DFT$ matrix. Then $\rm{KD}_{\mathcal{A},\mathcal{B}}^+=\rm ConvHull(\mathcal{A}\cup \mathcal{B}\cup \mathcal{C})$ when $d=p^2$.
\end{theorem}
\begin{proof}
To prove $\rm{KD}_{\mathcal{A},\mathcal{B}}^+= {\rm ConvHull}(\mathcal{A}\cup \mathcal{B}\cup \mathcal{C})$, we only need to show $\rm{KD}_{\mathcal{A},\mathcal{B}}^r = {\rm span}_{\mathbb{R}}({\mathcal{A}\cup \mathcal{B}\cup \mathcal{C}})$ by lemma \ref{lemma-3}.
Since ${\rm span}_{\mathbb{R}}({\mathcal{A}\cup \mathcal{B}\cup \mathcal{C}})\subseteq \rm{KD}_{\mathcal{A},\mathcal{B}}^r$, it means that we will obtain the desired result if $\dim \rm{KD}_{\mathcal{A},\mathcal{B}}^r= \dim {\rm span}_{\mathbb{R}}({\mathcal{A}\cup \mathcal{B}\cup \mathcal{C}})$. That is, we need to show $\dim {\rm{KD}}_{\mathcal{A},\mathcal{B}}^r=3p^2-2p$.

Suppose $F \in \rm{KD}_{\mathcal{A},\mathcal{B}}^r$. Next, we will apply lemma \ref{lemma-4} to discuss the dimension of $\rm{KD}_{\mathcal{A},\mathcal{B}}^r$.

By lemma \ref{lemma-4}, for off-diagonal entries in $F$, we consider two subcases gcd$(k,d)=1$ and gcd$(k,d)=p$.

(i) gcd$(k,d)=1$.  $k$ cannot take $p$ values of $0$, $p$, $2p$, $\cdots$, $(p-1)p$. Hence $k$ can take $p^2-p$ values. From the previous discussion about gcd$(k,d)=1$, we can see that the length of equation (\ref{F-k}) is $p^2$. By equation (\ref{F-k-(d-k)}), we have $F_{i(i+k)}= \overline {F_{i(i+(p^2-k))}}$. Thus $F_{i(i+k)}$ and $F_{i(i+(p^2-k))}$ are in a common category for all $i\in\mathbb Z_{p^2}$. Therefore, $F_{i(i+k)}$ for all $i\in\mathbb Z_{p^2}$ can be classified into $\frac{p^2-p}{2}$ categories.

(ii) gcd$(k,d)=p$, i.e., $k=p, 2p, \cdots, (p-1)p$. If $k=p$, then
\begin{equation}
\label{p2-F-gcd=p}
F_{(i-p)i}=F_{i(i+p)}=F_{(i+p)(i+2p)}=\cdots=F_{(i+(p-1)p)i}
\end{equation}
by equation (\ref{F-kt}). Notice that $i-p\equiv i+(p-1)p \mod p^2$ in equation (\ref{p2-F-gcd=p}). It implies the length of equation (\ref{p2-F-gcd=p}) is $p$.
Equation (\ref{p2-F-gcd=p}) means that $F_{i(i+p)}=F_{i'(i'+p)}$ when $i\equiv i' \mod p$, where $i'\in\mathbb Z_{p}$. Therefore, according to the value $i'$, $F_{i(i+p)}$ for all $i\in\mathbb Z_{p^2}$ can be classified into $p$ categories and each category has $p$ entries. By equation (\ref{F-k-(d-k)}), we have $F_{i(i+p)}= \overline {F_{i(i+(p^2-p))}}$. Thus $F_{i(i+p)}$ and $F_{i(i+(p^2-p))}$ for all $i\in\mathbb Z_{p^2}$ are in a common category. Thus, there are still $p$ categories and each category has $2p$ entries. Similarly, $F_{i(i+k)}$ and $F_{i(i+(p^2-k))}$ for all $i\in\mathbb Z_{p^2}$ can be classified into $p$ categories and each of which has $2p$ entries. Therefore, $F_{i(i+k)}$ for all $i\in\mathbb Z_{p^2}$ can be classified into $\frac{(p-1)p}{2}$ categories and each category has $2p$ entries.

Off-diagonal entries in $F$  can be classified into $\frac{p^2-p}{2}+\frac{(p-1)p}{2}=p^2-p$ categories. Since the entries in a common category are the same or conjugate, each category is determined by a nonreal value that is determined by two real parameters. Together with $p^2$ independent real diagonal entries, there are $2(p^2-p)+p^2=3p^2-2p$ real parameters. It means that $F$ is determined by $3p^2-2p$ independent real parameters. Therefore, $\dim {\rm{KD}}_{\mathcal{A},\mathcal{B}}^r= 3p^2-2p$.

\end{proof}

Langrenez et al.  \cite{Langrenez.2023} also showed that for prime $d$, a state $\rho\in{\rm ConvHull}(\mathcal{A}\cup \mathcal{B})$ if and only if \begin{equation}
\rho\in \rm{KD}_{\mathcal{A},\mathcal{B}}^+ \nonumber
\end{equation}
 and
\begin{equation}
\label{KD-Q}
Q_{ij}(\rho)+Q_{kl}(\rho)=Q_{il}(\rho)+Q_{kj}(\rho), \forall i,j,k,l\in\mathbb {Z}_d.
\end{equation}
They proved $\rm{KD}_{\mathcal{A},\mathcal{B}}^+={\rm ConvHull}(\mathcal{A}\cup \mathcal{B})$ by equation (\ref{KD-Q}). When $d=p^2$, we also obtain a similar result. It can be used to prove theorem \ref{theorem}. See \ref{Proof of Lemma 5} for details.

\section{Geometry of Kirkwood-Dirac classical states for  $d=pq$}
In this section, we try to prove $\rm{KD}_{\mathcal{A},\mathcal{B}}^+=$\rm ConvHull$(\mathcal{A}\cup \mathcal{B}\cup \mathcal{C}\cup \mathcal{D})$ for $d=pq$ and $p\neq q$. 
Referring to  the proof train of thought of  theorem \ref{theorem},
we  need to prove that $\rm span_{\mathbb{R}}(\mathcal{A}\cup \mathcal{B}\cup \mathcal{C}\cup \mathcal{D})={\rm KD}_{\mathcal{A},\mathcal{B}}^r$.

\begin{lemma}
\label{lemma-span=Vkdr}
If $d=pq$, then $\dim{\rm span}_{\mathbb{R}}(\mathcal{A}\cup \mathcal{B}\cup \mathcal{C}\cup \mathcal{D})=\dim{\rm KD}_{\mathcal{A},\mathcal{B}}^r=(2p-1)(2q-1)$. Furthermore, $\rm span_{\mathbb{R}}(\mathcal{A}\cup \mathcal{B}\cup \mathcal{C}\cup \mathcal{D})={\rm KD}_{\mathcal{A},\mathcal{B}}^r$.
\end{lemma}
The proof is given in \ref{Proof of Lemma-span=Vkdr}. This result completely characterizes KD real operators.

Next, we want to get an equivalent proposition similar to lemma \ref{lemma-3}, i.e., $ \rm span_{\mathbb{R}}({\mathcal{A}\cup \mathcal{B}\cup \mathcal{C}\cup \mathcal{D}})={KD}_{\mathcal{A},\mathcal{B}}^r $ if and only if $\rm{KD}_{\mathcal{A},\mathcal{B}}^+= ConvHull(\mathcal{A}\cup \mathcal{B}\cup \mathcal{C}\cup \mathcal{D})$ for $d=pq$. If the equivalent proposition holds, the desired result will be obtained. By the proof of lemma \ref{lemma-3}, we know that the proof of this equivalence proposition needs the condition
$\rm{KD}_{\mathcal{A},\mathcal{B}}^+\cap  span_{\mathbb{R}}(\mathcal{A}\cup \mathcal{B}\cup \mathcal{C}\cup \mathcal{D})= \rm ConvHull(\mathcal{A}\cup \mathcal{B}\cup \mathcal{C}\cup \mathcal{D})$. 
Unfortunately, we cannot prove $\rm{KD}_{\mathcal{A},\mathcal{B}}^+\cap  span_{\mathbb{R}}(\mathcal{A}\cup \mathcal{B}\cup \mathcal{C}\cup \mathcal{D})= \rm ConvHull(\mathcal{A}\cup \mathcal{B}\cup \mathcal{C}\cup \mathcal{D})$. Therefore, it is unclear whether  conjecture \ref{conjecture} is true or not for $d=pq$. However, through further exploration, we have the following result when $d=pq$.

\begin{theorem}
\label{theorem-pq}
If $d=pq$, then $\dim {\rm span}_{\mathbb{R}}(\mathcal{X}\cup \mathcal{Y}\cup \mathcal{Z})=3pq-p-q$ and $\rm{KD}_{\mathcal{A},\mathcal{B}}^+\cap  span_{\mathbb{R}}(\mathcal{X}\cup \mathcal{Y}\cup \mathcal{Z})= ConvHull(\mathcal{X}\cup \mathcal{Y}\cup \mathcal{Z})$, where $\mathcal{X}, \mathcal{Y},
\mathcal{Z}$  are three different sets randomly selected from  $\mathcal{A}, \mathcal{B}, \mathcal{C}$ and $\mathcal{D}$.
\end{theorem}
The proof of theorem \ref{theorem-pq} is similar to lemma \ref{lemma-1} and is given in \ref{Proof of theorem-pq}.

\section{Conclusions and discussions}
We have studied the KD classicality of quantum states based on DFT matrix. We  conjecture that $\rm{KD}_{\mathcal{A},\mathcal{B}}^+=\rm{ConvHull}(\rm {pure}({\rm {KD}_{\mathcal{A},\mathcal{B}}^+}))$ for general $d$ and show that this conjecture is true for the dimension $d=p^2$. That is, $\rm{KD}_{\mathcal{A},\mathcal{B}}^+={\rm ConvHull}(\mathcal{A}\cup\mathcal{B}\cup\mathcal{C})$ for $d=p^2$. This result is a generalization of the conclusion in reference \cite{Langrenez.2023}. For $d=pq$, although we do not completely characterize the set $\rm{KD}_{\mathcal{A},\mathcal{B}}^+$, some relevant results are given. We hope our results can lead to more findings in this field. It is worth exploring further whether this conjecture is true or not for general $d$.

A thorough understanding of KD classical states is a prerequisite for studying KD nonclassical states. Recently, KD nonclassicality measure based on the two quantum modification terms has been shown \cite{He-J.2024}. The measure refines the quantification of KD nonclassicality. In addition, there are several techniques that can be used to measure any KD distribution (regardless of it being described in terms of DFT)\cite{Quasiprobability.heat.2020,heat.flow.2024}.  It will be interesting to use these  techniques to measure the KD nonclassicality  based on DFT.

\ack
This work is supported by NSFC (Grant Nos. 11971151, 11901163, 62171264), the Fundamental Research Funds for the Universities of Henan Province (NSFRF220402).

\vspace{3mm}
\noindent\textbf{Data availability statement}

\vspace{2mm}

\noindent All data that support the findings of this study are included within the article (and any supplementary files)

\appendix
\section{The proof of Lemma \ref{lemma-3}}
\label{Proof of Lemma 3}
\begin{proof}
Let us prove sufficiency by contraposition. Suppose that ${\rm span}_{\mathbb{R}}({\mathcal{A}\cup \mathcal{B}\cup \mathcal{C}})\subset \rm{KD}_{\mathcal{A},\mathcal{B}}^r$ since ${\rm span}_{\mathbb{R}}(\mathcal{A}\cup \mathcal{B}\cup \mathcal{C})\subseteq
\rm{KD}_{\mathcal{A},\mathcal{B}}^r$. Thus
\begin{equation}
\rm{KD}_{\mathcal{A},\mathcal{B}}^r={\rm span}_{\mathbb{R}}({\mathcal{A}\cup \mathcal{B}\cup \mathcal{C}})\oplus W, \nonumber
\end{equation}
where $W$ is the nontrivial orthogonal complement space of span$_{\mathbb{R}}({\mathcal{A}\cup \mathcal{B}\cup \mathcal{C}})$ in $\rm{KD}_{\mathcal{A},\mathcal{B}}^r$. If $F\in W$, this implies $\langle a_{i}|F|a_{i}\rangle=0=\langle b_{j}|F|b_{j}\rangle$. It follows Tr$F = 0$. Let $F\in W\backslash\{0\}$. Let us consider
\begin{equation}
\label{rho=rho*+xF}
\rho(x)=\rho_{*}+xF,\ \forall x\in\mathbb{R},\rho_{*}=\frac{1}{d}\mathbb{I}_d\in {\rm ConvHull}({\mathcal{A}\cup \mathcal{B}\cup \mathcal{C}}).
\end{equation}
Hence, ${\rm Tr}\rho(x) = {\rm Tr}\rho_{*}+x{\rm Tr}F = 1$ for all $x\in\mathbb{R}$. Moreover, for any state $|\psi\rangle\in\mathcal{H}$, one has
\begin{equation}
\langle\psi|\rho(x)|\psi\rangle=\frac{1}{d}+x\langle\psi|F|\psi\rangle=\frac{1}{d}+xf_{i}\geq\frac{1}{d}-|x|f_{max},\ \forall x\in\mathbb{R}, \nonumber
\end{equation}
where $f_{i}$ are the eigenvalues of $F$ and $f_{max}=$max$\{|f_{i}|\mid i\in\mathbb Z_d\}>0$. If $x\in[-\frac{1}{df_{max}},\frac{1}{df_{max}}]$, then $\rho(x)$ is a density operator.

Since $F\in W\backslash\{0\}$, it follows $F\in \rm{KD}_{\mathcal{A},\mathcal{B}}^r$. Then $Q_{ij}(F)$ are real for all $i,j\in\mathbb Z_d$ and $\rho(x)\in \rm{KD}_{\mathcal{A},\mathcal{B}}^r$. One has
\begin{equation}
Q_{ij}(\rho(x))=\langle a_i|b_j\rangle\langle b_j|\rho(x)|a_i\rangle
=\frac{1}{d}+xQ_{ij}(F)\geq\frac{1}{d}-|x|\textrm{max}_{i,j}|Q_{ij}(F)|,\ \forall x\in\mathbb{R}. \nonumber
\end{equation}
If $x\in[-\frac{1}{d\textrm{max}_{i,j}|Q_{ij}(F)|},\frac{1}{d\textrm{max}_{i,j}|Q_{ij}(F)|}]$, then $Q_{ij}(\rho(x))\geq0$ for all $i,j\in\mathbb Z_d$.
Let \begin{equation}
x_+={\rm min}\{\frac{1}{df_{max}},\frac{1}{d\textrm{max}_{i,j}|Q_{ij}(F)|}\}>0.\nonumber
\end{equation}  Thus $\rho(x)\in\rm{KD}_{\mathcal{A},\mathcal{B}}^+$ when $x\in[-x_+,x_+]$. That $\rho(x)\notin {\rm span}_{\mathbb{R}}({\mathcal{A}\cup \mathcal{B}\cup \mathcal{C}})$ for all $x\neq0$ by equation (\ref{rho=rho*+xF}). It is obvious that ${\rm ConvHull}({\mathcal{A}\cup \mathcal{B}\cup \mathcal{C}})\subseteq $span$_{\mathbb{R}}({\mathcal{A}\cup \mathcal{B}\cup \mathcal{C}})$, then $\rho(x)\notin {\rm ConvHull}({\mathcal{A}\cup \mathcal{B}\cup \mathcal{C}})$ for all $x\neq0$. It follows $\rm{KD}_{\mathcal{A},\mathcal{B}}^+\neq{\rm ConvHull}({\mathcal{A}\cup \mathcal{B}\cup \mathcal{C}})$, which is contradictory to $\rm{KD}_{\mathcal{A},\mathcal{B}}^+= {\rm ConvHull}({\mathcal{A}\cup \mathcal{B}\cup \mathcal{C}})$. This completes the proof.
\end{proof}

\section{The proof of Lemma \ref{lemma-4}}
\label{The proof of lemma 4}
\begin{proof}
The entries of DFT matrix $U$ are $U_{ij}=\langle a_i|b_j\rangle=\frac{1}{\sqrt{d}}\omega_{d}^{ij}$, where $i,j\in\mathbb{Z}_{d}$. For simplicity, the notation of $\omega_{d}^{ij}$ is written as $\omega^{ij}$ in this proof.

 By the definition of $\rm{KD}_{\mathcal{A},\mathcal{B}}^r$, we can know that   $F\in \rm{KD}_{\mathcal{A},\mathcal{B}}^r$ if and only if ${\rm Im}(Q_{ij}(F))=0$ for all $i,j\in\mathbb Z_d$, where ${\rm Im}(Q_{ij}(F))$ is the imaginary part of $Q_{ij}(F)$. First of all,
\begin{equation}
\label{Qij(F)}
Q_{ij}(F)=\sum_{k=0}^{d-1}\langle b_j|a_i\rangle\langle a_i|F|a_k\rangle\langle a_k|b_j\rangle =\frac{1}{d}\sum_{k=0}^{d-1}\omega^{(k-i)j}F_{ik},\  \forall \ i,j \in\mathbb Z_d.
\end{equation}
Let $k'=k-i\in \mathbb Z_d$, then equation (\ref{Qij(F)}) can be written as
\begin{equation}
\label{k‘-Qij(F)}
Q_{ij}(F)=\frac{1}{d}\sum_{k'=0}^{d-1}\omega^{k'j}F_{i(i+k')}=\frac{1}{d}\sum_{k'=1}^{d-1}\omega^{k'j}F_{i(i+k')}+\frac{1}{d}F_{ii},\  \forall \ i,j \in\mathbb Z_d.
\end{equation}
Since $F$ is self-adjoint, it follows
\begin{eqnarray}
\label{k-i Im}
{\rm Im}(Q_{ij}(F))
=\frac{1}{2d\sqrt{-1}}\sum_{k'=1}^{d-1}(\omega^{k'j}F_{i(i+k')}-\overline{\omega}^{k'j}F_{(i+k')i}), \  \forall \ i,j \in\mathbb Z_d.
\end{eqnarray}

We divide the proof into two cases, odd $d$ and even $d$.

Case 1. Odd $d$. If $k'\in\{1,\ldots,\frac{d-1}{2}\}$, then $d-k'\in\{\frac{d+1}{2},\ldots,d-1\}$. In equation (\ref{k-i Im}),
\begin{eqnarray}
&&(\omega^{k'j}F_{i(i+k')}-\overline{\omega}^{k'j}F_{(i+k')i})+(\omega^{(d-k')j}F_{i(i+(d-k'))}-\overline{\omega}^{(d-k')j}F_{(i+(d-k'))i}) \nonumber\\
&=&(\omega^{k'j}F_{i(i+k')}-\overline{\omega}^{k'j}F_{(i+k')i})+(\overline\omega^{k'j}F_{i(i-k')}-\omega^{k'j}F_{(i-k')i}) \nonumber\\
&=&\omega^{k'j}(F_{i(i+k')}-F_{(i-k')i})+\overline\omega^{k'j}(F_{i(i-k')}-F_{(i+k')i}).\nonumber
\end{eqnarray}
Let $z_{k'}:=F_{i(i+k')}-F_{(i-k')i}$. Then equation (\ref{k-i Im}) can be rewritten as
\begin{eqnarray}
\label{eq-ImQ-odd}
{\rm Im}(Q_{ij}(F))=\frac{1}{2d\sqrt{-1}}(\sum_{k'=1}^{\frac{d-1}{2}}\omega^{k'j}z_{k'}+\overline{\omega}^{k'j}(-\overline z_{k'})),\  \forall \ i,j \in\mathbb Z_d.
\end{eqnarray}
Recall $F\in \rm{KD}_{\mathcal{A},\mathcal{B}}^r$ if and only if ${\rm Im}(Q_{ij}(F))=0$ for all $i,j \in\mathbb Z_d$. If $F\in \rm{KD}_{\mathcal{A},\mathcal{B}}^r$, then
\begin{eqnarray}
\label{equation{k'j}}
\sum_{k'=1}^{\frac{d-1}{2}}\omega^{k'j}z_{k'}+\overline{\omega}^{k'j}(-\overline z_{k'})=0,\  \forall \ j\in\mathbb Z_{d}.
\end{eqnarray}
This equation (\ref{equation{k'j}}) can be written as
\begin{equation}
Mz=0,\nonumber
\end{equation}
where $z=(z_{1},\ldots,z_{\frac{d-1}{2}},-\overline z_{\frac{d-1}{2}},\ldots,-\overline z_{1})^{T}$, $``T"$ denotes the transpose of a matrix, and
\begin{eqnarray}
M&=&
\begin{pmatrix}
{\omega} & \omega^{2} & \cdots & \omega^{\frac{d-1}{2}} & \overline \omega^{\frac{d-1}{2}} & \cdots & \overline \omega^{2} & \overline \omega^{1}\\
\omega^{2} & \omega^{4} & \cdots & \omega^{d-1} & \overline \omega^{d-1} & \cdots & \overline \omega^{4} & \overline \omega^{2}\\
\vdots & \vdots & \vdots & \vdots & \vdots & \vdots & \vdots & \vdots\\
\omega^{d-1} & \omega^{2(d-1)} & \cdots & \omega^{\frac{(d-1)^2}{2}} & \overline \omega^{\frac{(d-1)^2}{2}} & \cdots & \overline \omega^{2(d-1)} & \overline \omega^{d-1}\\
\end{pmatrix} \nonumber\\
&=&\begin{pmatrix}
\omega & \omega^{2} & \cdots & \omega^{\frac{d-1}{2}} & \omega^{\frac{d+1}{2}} & \cdots & \omega^{d-2} & \omega^{d-1}\\
\omega^{2} & \omega^{4} & \cdots & \omega^{d-1} &  \omega^{d+1} & \cdots & \omega^{2(d-2)} & \omega^{2(d-1)}\\
\vdots & \vdots & \vdots & \vdots & \vdots & \vdots & \vdots & \vdots\\
\omega^{d-1} & \omega^{2(d-1)} & \cdots & \omega^{\frac{(d-1)^2}{2}} & \omega^{\frac{(d-1)(d+1)}{2}} & \cdots & \omega^{(d-1)(d-2)} & \omega^{(d-1)^2}\\
\end{pmatrix}.\nonumber
\end{eqnarray}
Since the entries in the first row of $M$ are all distinct, according to the property of Vandermonde matrix, this matrix $M$ is invertible. It follows that $z_{k'}=0$ for all $k'\in\{1,\ldots,\frac{d-1}{2}\}$, i.e.,
\begin{equation}
\label{F}
F_{i(i+k')}=F_{(i-k')i},\forall \ i\in\mathbb Z_{d}, k'\in\mathbb Z^*_{\frac{d-1}{2}},
\end{equation}
where $\mathbb Z^*_{\frac{d-1}{2}}=\mathbb Z_{\frac{d-1}{2}}\backslash \{0\}$. If $k'\in\{\frac{d+1}{2},\ldots,d-1\}$, then $d-k'\in\{1,\ldots,\frac{d-1}{2}\}$. Substitute $d-k'$ into equation (\ref{F}), and then $F_{i(i-k')}=F_{(i+k')i}$ for all $ i\in\mathbb Z_{d}$ and $k'\in\{\frac{d+1}{2},\ldots,d-1\}$. Since $F$ is self-adjoint, it follows $F_{i(i+k')}=\overline{F_{(i+k')i}}$ and $F_{(i-k')i}=\overline{F_{i(i-k')}}$. It follows
\begin{equation}
F_{i(i+k')}=F_{(i-k')i},\ \forall\ i\in\mathbb Z_{d},\ k'\in\{\frac{d+1}{2},\ldots,d-1\}.\nonumber
\end{equation}
It is obvious that $F_{ii}=F_{ii}$ when $k'=0$. From the above, one obtains
\begin{equation}
F_{i(i+k')}=F_{(i-k')i},\ \forall\ i, k'\in\mathbb Z_d.\nonumber
\end{equation}

Conversely, if $F_{i(i+k')}=F_{(i-k')i}$ for all $i, k'\in\mathbb Z_d$, then $z_{k'}=0$ and $-\overline z_{k'}=0$. That ${\rm Im}(Q_{ij}(F))=0$ for all $i, j\in\mathbb Z_d$ can be obtained from equation (\ref{eq-ImQ-odd}). Thus $F \in \rm{KD}_{\mathcal{A},\mathcal{B}}^r$.

Case 2. Even $d$. Its proof is similar to that of Case 1. If $k'\in\{1,\ldots,\frac{d}{2}-1\}$, then $d-k'\in\{\frac{d}{2}+1,\ldots,d-1\}$. In equation (\ref{k-i Im}),
\begin{eqnarray}
&&(\omega^{k'j}F_{i(i+k')}-\overline{\omega}^{k'j}F_{(i+k')i})+(\omega^{(d-k')j}F_{i(i+(d-k'))}-\overline{\omega}^{(d-k')j}F_{(i+(d-k'))i}) \nonumber\\
&=&(\omega^{k'j}F_{i(i+k')}-\overline{\omega}^{k'j}F_{(i+k')i})+(\overline\omega^{k'j}F_{i(i-k')}-\omega^{k'j}F_{(i-k')i}) \nonumber\\
&=&\omega^{k'j}(F_{i(i+k')}-F_{(i-k')i})+\overline\omega^{k'j}(F_{i(i-k')}-F_{(i+k')i}).\nonumber
\end{eqnarray}
Let $z_{k'}:=F_{i(i+k')}-F_{(i-k')i}$. Then equation (\ref{k-i Im}) can be written as
\begin{equation}
\label{eq-ImQ-even}
{\rm Im}(Q_{ij}(F))=\frac{1}{2d\sqrt{-1}}(\sum_{k'=1}^{\frac{d}{2}-1}\omega^{k'j}z_{k'}+\overline{\omega}^{k'j}(-\overline z_{k'})+\omega^{\frac{d}{2}j}z_{\frac{d}{2}}),\ \forall \ i,j \in\mathbb Z_d.\nonumber
\end{equation}
Recall F $\in \rm{KD}_{\mathcal{A},\mathcal{B}}^r$ if and only if ${\rm Im}(Q_{ij}(F))=0$ for all $i,j \in\mathbb Z_d$. If $\in \rm{KD}_{\mathcal{A},\mathcal{B}}^r$, then
\begin{eqnarray}
\label{even-mz=0}
\sum_{k'=1}^{\frac{d-1}{2}}(\omega^{k'j}z_{k'}+\overline{\omega}^{k'j}(-\overline z_{k'})+\omega^{\frac{d}{2}j}z_{\frac{d}{2}})=0,\ \forall \ j \in\mathbb Z_{d}.
\end{eqnarray}
This equation (\ref{even-mz=0}) can be written as
\begin{equation}
Mz=0,
\nonumber
\end{equation}
where $z=(z_{1},\ldots,z_{\frac{d}{2}-1},z_{\frac{d}{2}},-\overline z_{\frac{d}{2}-1},\ldots,-\overline z_{1})^T$, $``T"$ is denoted the transpose of a matrix, and
\begin{eqnarray}
M&=&\begin{pmatrix}
\omega & \cdots & \omega^{\frac{d}{2}-1} & \omega_{\frac{d}{2}} & \overline \omega^{\frac{d}{2}-1} & \cdots  & \overline \omega\\
\omega^{2}  & \cdots & \omega^{2(\frac{d}{2}-1)} & \omega^{2\frac{d}{2}} & \overline \omega^{2(\frac{d}{2}-1)} & \cdots  & \overline \omega^{2}\\
\vdots &  \vdots & \vdots & \vdots & \vdots & \vdots &  \vdots\\
\omega^{d-1}  & \cdots & \omega^{(d-1)(\frac{d}{2}-1)} & \omega^{(d-1)(\frac{d}{2})} & \overline \omega^{(d-1)(\frac{d}{2}-1)} & \cdots  & \overline \omega^{d-1}\\
\end{pmatrix} \nonumber\\
&=&\begin{pmatrix}
\omega  & \cdots & \omega^{\frac{d}{2}-1} & \omega_{\frac{d}{2}} & \omega^{\frac{d}{2}+1} & \cdots & \omega^{d-1}\\
\omega^{2}  & \cdots & \omega^{2(\frac{d}{2}-1)} & \omega^{2(\frac{d}{2})} & \omega^{2(\frac{d}{2}+1)} & \cdots  & \omega^{2(d-1)}\\
\vdots & \vdots & \vdots & \vdots & \vdots & \vdots  & \vdots\\
\omega^{d-1}  & \cdots & \omega^{(d-1)(\frac{d}{2}-1)} & \omega^{(d-1)(\frac{d}{2})}& \omega^{(d-1)(\frac{d}{2}+1)} & \cdots  & \omega^{(d-1)^2}\\
\end{pmatrix}\nonumber
\end{eqnarray}
Since the entries in the first row are all different, according to the property of
Vandermonde matrix, this matrix $M$ is invertible. It follows that $z_{k'}=0$ for all $k'\in\{1,\ldots,\frac{d}{2}-1\}$, i.e.,
\begin{equation}
\label{eq-F-ik}
F_{i(i+k')}=F_{(i-k')i},\ \forall\ i\in\mathbb Z_d,\ k'\in\mathbb Z^*_{\frac{d}{2}-1}.
\end{equation}
If $k'\in\{\frac{d}{2}+1,\ldots,d-1\}$, then $d-k'\in\{1,\ldots,\frac{d}{2}-1\}$.  Substituting $d-k'$ into equation (\ref{eq-F-ik}), one can obtain $F_{i(i-k')}=F_{(i+k')i}$ for all $i\in\mathbb Z_d$ and $k'\in\{\frac{d}{2}+1,\ldots,d-1\}$. Since F is self-adjoint, it follows $F_{i(i+k')}=\overline{F_{(i+k')i}}$ and $F_{(i-k')i}=\overline{F_{i(i-k')}}$. Therefore,
\begin{equation}
F_{i(i+k')}=F_{(i-k')i},\ \forall\ i\in\mathbb Z_d,\ k'\in\{\frac{d}{2}+1,\ldots,d-1\}. \nonumber
\end{equation}
It is obvious that $F_{ii}=F_{ii}$ when $k'=0$. From the above, one has
\begin{equation}
F_{i(i+k')}=F_{(i-k')i},\ \forall\ i, k'\in\mathbb Z_d.\nonumber
\end{equation}
Conversely, if $F_{i(i+k')}=F_{(i-k')i}$ for all $i, k'\in\mathbb Z_d$, then $z_{k'}=0$ and $-\overline z_{k'}=0$. That ${\rm Im}(Q_{ij}(F))=0$ for all $i, j\in\mathbb Z_d$ can be obtained from equation (\ref{eq-ImQ-even}). Thus $F \in \rm{KD}_{\mathcal{A},\mathcal{B}}^r$.
\end{proof}

\section{Another method to prove Theorem \ref{theorem}}
\label{Proof of Lemma 5}
We first give a lemma, which can also prove Theorem \ref{theorem}.
\begin{lemma}
\label{lemma-5}
If $d=p^2$, then
\begin{equation}
\rho\in \rm ConvHull(\mathcal{A}\cup \mathcal{B}\cup \mathcal{C}) \nonumber
\end{equation}
if and only if the following three statements hold,
\begin{flalign*}
\ \ \ \ \ & (i)\ \rho\in\rm{KD}_{\mathcal{A},\mathcal{B}}^+, &\\
\ \ \ \ \ &(ii)\ Q_{ij}(\rho)+Q_{kl}(\rho)=Q_{il}(\rho)+Q_{kj}(\rho),\forall\ i \equiv k \mod p,\  j,l\in\mathbb Z_d,\\
\ \ \ \ \ &(iii)\ Q_{ij}(\rho)+Q_{kl}(\rho)=Q_{il}(\rho)+Q_{kj}(\rho),\forall\ j \equiv l \mod p,\  i,k\in\mathbb Z_d.&
\end{flalign*}
\end{lemma}
\begin{proof}
Let us first consider sufficiency. Since $\rho\in\rm{KD}_{\mathcal{A},\mathcal{B}}^+$, it follows $Q_{ij}(\rho)\geq0$ for all $i,j\in\mathbb Z_d$. For $\forall \ m,s\in\mathbb Z_p$, suppose
${\rm min}_{i,j\in\mathbb Z_d}\{ \frac {Q_{ij}(\rho)}{{|\langle a_i|b_j\rangle|}^2}|{i\equiv m\mod p,j\equiv s\mod p} \} =: \frac{Q_{ms}(\rho)}{{|\langle a_m|b_s\rangle|}^2}$. That ${|\langle a_i|b_j\rangle|}^2=\frac{1}{d}$ for all $i,j\in\mathbb Z_d$ since $U$ is the DFT matrix.
Define \begin{equation}
\lambda_i =d(Q_{is}(\rho)-Q_{ms}(\rho)), \mu_j = d(Q_{mj}(\rho)-Q_{ms}(\rho)), \gamma_{ms} = dQ_{ms}(\rho),\nonumber
\end{equation}
where $i \equiv m \mod p$, $j \equiv s \mod p$. That $Q_{ij}(\rho)-Q_{kj}(\rho)=Q_{il}(\rho)-Q_{kl}(\rho)$ holds for all $i \equiv k\mod p, j,l\in\mathbb Z_d$ by the second statement. It implies
$\lambda_i$ is well defined independent of $s$. Similarly, by the third statement, $\mu_j$ is well defined independent of $m$.
Construct a state
\begin{equation}
\rho = \sum_{i=0}^{d-1}\lambda_ia_i+\sum_{j=0}^{d-1}\mu_jb_j+\sum_{m,s=0}^{p-1}\gamma_{ms}\psi_{ms}. \nonumber
\end{equation}
Notice that $\lambda_i, \mu_j, \gamma_{ms}\geq 0$ for all $i, j\in\mathbb Z_d$ and $m, s\in\mathbb Z_p$.

Next, we will show Tr$\rho=1$.
\begin{eqnarray}
\label{Tr-rho}
{\rm Tr}\rho &=& \sum_{i=0}^{d-1}\lambda_i + \sum_{j=0}^{d-1}\mu_j + \sum_{m,s=0}^{p-1}\gamma_{ms} \nonumber\\
&=& \frac{1}{d}\sum_{i,j=0}^{d-1}\lambda_i + \frac{1}{d}\sum_{i,j=0}^{d-1}\mu_j + d\sum_{m,s=0}^{p-1}Q_{ms}(\rho) \nonumber\\
&=& \sum_{m,k',s,l'=0}^{p-1}(Q_{(m+k'p)s}(\rho)-Q_{ms}(\rho)) + \sum_{m,k',s,l'=0}^{p-1}(Q_{m(s+l'p)}(\rho)-Q_{ms}(\rho))\nonumber\\
&\;&+ \sum_{m,k',s,l'=0}^{p-1}Q_{ms}(\rho)  \nonumber\\
&=& \sum_{m,k',s,l'=0}^{p-1}(Q_{(m+k'p)s}(\rho) + Q_{m(s+l'p)}(\rho) - Q_{ms}(\rho)) \nonumber\\
&=& \sum_{m,k',s,l'=0}^{p-1}Q_{(m+k'p)(s+l'p)}(\rho) = 1.
\end{eqnarray}
The third equality  in equation (\ref{Tr-rho}) holds since $i=m+k'p, j=s+l'p$ for all $m,k',s,l'\in\mathbb Z_p$. The fourth equality  in equation (\ref{Tr-rho}) holds since $Q_{ij}(\rho)+Q_{kl}(\rho)=Q_{il}(\rho)+Q_{kj}(\rho)$, for $\forall\ i \equiv k \mod p, \forall\ j \equiv l\mod p$.
Thus $\rho\in {\rm ConvHull}(\mathcal{A}\cup \mathcal{B}\cup \mathcal{C})$.

Secondly, we prove the necessity. Since $\rho\in {\rm ConvHull}(\mathcal{A}\cup \mathcal{B}\cup \mathcal{C})$, one has
\begin{equation}
\rho = \sum_{i=0}^{d-1}\lambda_ia_i+\sum_{j=0}^{d-1}\mu_jb_j+\sum_{m,s=0}^{p-1}\gamma_{ms}\psi_{ms}, \nonumber
\end{equation}
where $\sum_{i=0}^{d-1}\lambda_i + \sum_{j=0}^{d-1}\mu_j + \sum_{m,s=0}^{p-1}\gamma_{ms} = 1$ and $\lambda_i, \mu_j, \gamma_{ms}\geq 0$ for all $i, j\in\mathbb Z_d$, $m, s\in\mathbb Z_p$.
Therefore,
\begin{equation}
Q_{ij}(\rho) = \frac{1}{d}(\lambda_i+\mu_j+\gamma_{ms})\geq0, \forall\ i, j\in\mathbb Z_d, \nonumber
\end{equation}
where $i \equiv m \mod p$, $j \equiv s \mod p$. Then $\rho\in\rm{KD}_{\mathcal{A},\mathcal{B}}^+$.

Next, we will prove the second statement. For $i \equiv k\mod p$ and $j,l\in\mathbb Z_d$, suppose $i, k \equiv m\mod p$, $j \equiv s \mod p $, $l \equiv s' \mod p $ with $m, s, s'\in\mathbb Z_p$. It follows
\begin{equation}
\begin{split}
Q_{ij}(\rho) &= \frac{1}{d}(\lambda_i+\mu_j+\gamma_{ms}), \nonumber \\
Q_{kl}(\rho)&= \frac{1}{d}(\lambda_k+\mu_l+\gamma_{ms'}), \nonumber \\
Q_{il}(\rho)&= \frac{1}{d}(\lambda_i+\mu_l+\gamma_{ms'}), \nonumber \\
Q_{kj}(\rho)&=\frac{1}{d}(\lambda_k+\mu_j+\gamma_{ms}). \nonumber
\end{split}
\end{equation}
Thus
\begin{equation}
Q_{ij}(\rho)+Q_{kl}(\rho)=\frac{1}{d}(\lambda_i+\mu_j+\gamma_{ms}+\lambda_k+\mu_l+\gamma_{ms'})=Q_{il}(\rho)+Q_{kj}(\rho). \nonumber
\end{equation}
Similarly, we can prove that the third statement still holds.
That completes the proof.
\end{proof}

Next, we will apply lemma \ref{lemma-5} to prove  theorem \ref{theorem}. Since ${\rm ConvHull}(\mathcal{A}\cup \mathcal{B}\cup \mathcal{C})\subseteq \rm{KD}_{\mathcal{A},\mathcal{B}}^+$, we only need to prove that $\rm{KD}_{\mathcal{A},\mathcal{B}}^{+}\subseteq{\rm ConvHull}(\mathcal{A}\cup \mathcal{B}\cup \mathcal{C})$. Therefore, we need to prove the second statement and the third statement. To prove the third statement, we need to use a corollary from lemma \ref{lemma-4}.

\begin{corollary}
\label{corollary}
For any dimension $d$, a self-adjoint operator $G\in \rm{KD}_{\mathcal{A},\mathcal{B}}^r$ if and only if
\begin{equation}
\label{eq-G}
G_{j(j+k)}=G_{(j-k)j}, \forall \ j,k \in\mathbb Z_d,
\end{equation}
where $G_{kj}=\langle b_k|G|b_j\rangle$.
\end{corollary}
The proof of this corollary is similar to that of lemma \ref{lemma-4}. Here we will not repeat it. Next, we give another method to prove  theorem \ref{theorem}.

\begin{proof}
Let us first show the second statement of lemma \ref{lemma-5} holds.
Since $\rho\in\rm{KD}_{\mathcal{A},\mathcal{B}}^+$, obviously, $\rho\in \rm{KD}_{\mathcal{A},\mathcal{B}}^r$. From equation (\ref{k‘-Qij(F)}), one has
\begin{equation}
\label{Qij-rho}
Q_{ij}(\rho)=\frac{1}{d}\sum_{k'=1}^{d-1}{\rm Re}(\omega^{k'j}\rho_{i(i+k')})+\frac{1}{d}\rho_{ii}, \  \forall \ i,j \in\mathbb Z_d.
\end{equation}
It follows
\begin{equation}
\label{eq-Qij-kl}
Q_{ij}(\rho)+Q_{kl}(\rho)=\frac{1}{d}\sum_{k'=1}^{d-1}{\rm Re}(\omega^{k'j}\rho_{i(i+k')})+\frac{1}{d}\rho_{ii}+\frac{1}{d}\sum_{k'=1}^{d-1}{\rm Re}(\omega^{k'l}\rho_{k(k+k')})+\frac{1}{d}\rho_{kk},\nonumber
\end{equation}
\begin{equation}
\label{eq-Qkj-il}
Q_{kj}(\rho)+Q_{il}(\rho)=\frac{1}{d}\sum_{k'=1}^{d-1}{\rm Re}(\omega^{k'j}\rho_{k(k+k')})+\frac{1}{d}\rho_{kk}+\frac{1}{d}\sum_{k'=1}^{d-1}{\rm Re}(\omega^{k'l}\rho_{i(i+k')})+\frac{1}{d}\rho_{ii}.\nonumber
\end{equation}

It is easy to verify that if $\rho_{i(i+k')}=\rho_{k(k+k')}$ for all $i,k\in\mathbb Z_{d}$, $i\equiv k \mod p$ and $k'\in\mathbb Z^*_{d}$, then the second statement holds.
Next, we will prove $\rho_{i(i+k')}=\rho_{k(k+k')}$ by two cases, gcd$(k',d)=p$ and gcd$(k',d)=1$.

Case 1. gcd$(k',d)=p$. Since $\rho\in \rm{KD}_{\mathcal{A},\mathcal{B}}^r$, by equation (\ref{F-kt}), one has
\begin{equation}
\label{rho-i-i+gk'}
\rho_{i(i+k')}=\rho_{(i+gk')(i+(g+1)k')}, \ \forall\  i\in\mathbb Z_{d},\ g\in\mathbb Z_{p}.
\end{equation}
Let $k:=i+gk'$. It follows $i\equiv k \mod p$. Then equation (\ref{rho-i-i+gk'}) can be written as $\rho_{i(i+k')}=\rho_{k(k+k')}$.

Case 2. gcd$(k',d)=1$. Since $\rho\in \rm{KD}_{\mathcal{A},\mathcal{B}}^r$, by equation (\ref{F-k}), one has
\begin{equation}
\label{rho-i-i+hk'}
\rho_{i(i+k')}=\rho_{(i+hk')(i+(h+1)k')}, \ \forall\  i,h\in\mathbb Z_{d}.
\end{equation}
Similarly, let $k:=i+hk'$. Then equation (\ref{rho-i-i+hk'}) can be written as $\rho_{i(i+k')}=\rho_{k(k+k')}$ for $\forall\ i, k\in\mathbb Z_{d}$. Specially, for $i\equiv k \mod p$ this equation still holds.
The desired result is obtained.

Secondly, let us show the third statement of lemma \ref{lemma-5} holds. Since $\rho\in \rm{KD}_{\mathcal{A},\mathcal{B}}^r$, it follows \begin{equation}
Q_{ij}(\rho)=\sum_{k=0}^{d-1}{\rm Re}(\langle b_j|a_i\rangle\langle a_i|b_k\rangle \langle b_k|\rho|b_j\rangle) =\frac{1}{d}\sum_{k=0}^{d-1}{\rm Re}(\omega^{i(k-j)}\rho_{kj}),\  \forall \ i,j \in\mathbb Z_d.\nonumber
\end{equation}
Let $k'=k-j\in \mathbb Z_d$, one can obtain the following equation that is similar to equation (\ref{Qij-rho}),
\begin{equation}
\label{k‘-Qij(G)}
Q_{ij}(\rho)=\frac{1}{d}\sum_{k'=1}^{d-1}{\rm Re}(\omega^{ik'}\rho_{(j+k')j})+\frac{1}{d}\rho_{jj}, \  \forall \ i,j \in\mathbb Z_d. \nonumber
\end{equation}
It follows
\begin{equation}
\label{G-eq-Qij-kl}
Q_{ij}(\rho)+Q_{kl}(\rho)=\frac{1}{d}\sum_{k'=1}^{d-1}{\rm Re}(\omega^{ik'}\rho_{(j+k')j})+\frac{1}{d}\rho_{jj}+\frac{1}{d}\sum_{k'=1}^{d-1}{\rm Re}(\omega^{kk'}\rho_{(l+k')l})+\frac{1}{d}\rho_{ll},\nonumber
\end{equation}
\begin{equation}
\label{G-eq-Qkj-il}
Q_{kj}(\rho)+Q_{il}(\rho)=\frac{1}{d}\sum_{k'=1}^{d-1}{\rm Re}(\omega^{kk'}\rho_{(j+k')j})+\frac{1}{d}\rho_{jj}+\frac{1}{d}\sum_{k'=1}^{d-1}{\rm Re}(\omega^{ik'}\rho_{(l+k')l})+\frac{1}{d}\rho_{ll}.\nonumber
\end{equation}
Similar to the previous processes, we can show $\rho_{(j+k')j}=\rho_{(l+k')l}$ for all $i,k\in\mathbb Z_{d}$, $j\equiv l\mod p$ and $k'\in\mathbb Z^*_{d}$ by equation (\ref{eq-G}).
Therefore, the third statement holds.
Then KD$^{+}_{\mathcal{A},\mathcal{B}}\subseteq{\rm ConvHull}(\mathcal{A}\cup \mathcal{B}\cup \mathcal{C})$ by lemma \ref{lemma-5}. The proof has been completed.
\end{proof}

\section{The proof of Lemma \ref{lemma-span=Vkdr}}
\label{Proof of Lemma-span=Vkdr}
\begin{proof}
Let us first show the dimension of $\rm span_{\mathbb{R}}(\mathcal{A}\cup \mathcal{B}\cup \mathcal{C}\cup \mathcal{D})$. The  proof  of  $\dim \rm span_{\mathbb{R}}(\mathcal{A}\cup \mathcal{B}\cup \mathcal{C}\cup \mathcal{D})$ is similar to the proof of the first statement of lemma \ref{lemma-1}. Define a linear map $\Gamma:\mathbb{R}^{4pq}\rightarrow \rm{\rm span}_{\mathbb{R}}(\mathcal{A}\cup \mathcal{B}\cup \mathcal{C}\cup \mathcal{D})$, for each $((\lambda_i),(\mu_j),(\gamma_{ms}),(\eta_{m's'}))\in\mathbb{R}^{4pq}$
\begin{equation}
\Gamma((\lambda_i),(\mu_j),(\gamma_{ms}),(\eta_{m's'}))=\sum_{i=0}^{pq-1}\lambda_ia_i+\sum_{j=0}^{pq-1}\mu_jb_j+\sum_{{m\in\mathbb Z_p}\atop{s\in\mathbb Z_q}}\gamma_{ms}\psi_{ms}+\sum_{{m'\in\mathbb Z_q}\atop{s'\in\mathbb Z_p}}\eta_{m's'}\varphi_{m's'}.\nonumber
\end{equation}
From the rank theorem, we have $\dim\rm{\rm span}_{\mathbb{R}}(\mathcal{A}\cup \mathcal{B}\cup \mathcal{C}\cup \mathcal{D})=$ $4pq-\dim(\rm{Ker} (\Gamma))$.

Suppose
\begin{equation}
\Upsilon=\sum_{i=0}^{pq-1}\lambda_ia_i+\sum_{j=0}^{pq-1}\mu_jb_j+\sum_{{m\in\mathbb Z_p}\atop{s\in\mathbb Z_q}}\gamma_{ms}\psi_{ms}+\sum_{{m'\in\mathbb Z_q}\atop{s'\in\mathbb Z_p}}\eta_{m's'}\varphi_{m's'}=0, \nonumber
\end{equation}
where $\lambda_i,\mu_j,\gamma_{ms},\eta_{m's'}\in \mathbb{R}$. We will calculate $\dim({\rm Ker}(\Gamma))$ below.
Since $\Upsilon=0$, one has
\begin{equation}
\langle a_i|\Upsilon|b_j\rangle=\langle a_i|b_j\rangle(\lambda_i+\mu_j+\gamma_{ms}+\eta_{m's'}) = 0,\ i,j\in\mathbb Z_{pq},\nonumber
\end{equation}
where
\begin{eqnarray}
\left\{
\begin{aligned}
i \equiv m \mod p \\
i \equiv m' \mod q
\end{aligned}
\right.,
\ \ \
\left\{
\begin{aligned}
j \equiv s \mod q \\
j \equiv s' \mod p
\end{aligned}
\right..\nonumber
\end{eqnarray}
According to Chinese remainder theorem, $i$ is uniquely determined in $\mathbb Z_{pq}$ by $m$ and $m'$, and $j$ is uniquely determined in $\mathbb Z_{pq}$ by $s$ and $s'$. Since $\langle a_i|b_j\rangle=\frac{1}{\sqrt{d}}\omega_d^{ij} \neq 0$, one has $\lambda_i+\mu_j+\gamma_{ms}+\eta_{m's'} = 0$, i.e.,
\begin{equation}
\label{eq-mm'}
\lambda_{mm'}+\mu_{ss'}+\gamma_{ms}+\eta_{m's'} = 0,  \ \textrm{for}\  \forall\ m,s'\in\mathbb Z_p,\ m',s\in\mathbb Z_q.
\end{equation}
It follows
\begin{eqnarray}
\lambda_{mm'}+\mu_{s0}+\gamma_{ms} + \eta_{m'0} = \lambda_{mm'}+\mu_{s{s}'}+\gamma_{ms}+\eta_{m'{s}'},\nonumber
\end{eqnarray}
for all $m,s'\in\mathbb Z_p$, $m',s\in\mathbb Z_q$. Thus
\begin{eqnarray}
\label{s1'=0}
\eta_{m'{s}'}-\eta_{m'0} =\mu_{s0}-\mu_{s{s}'},\ \textrm{for}\ \forall\ {s}'\in\mathbb Z_p,\ m',s\in\mathbb Z_q.
\end{eqnarray}
It follows
\begin{eqnarray}
\label{eta=mu+eta-m'0}
\eta_{m'{s}'} =\mu_{s0}-\mu_{s{s}'}+\eta_{m'0},\ \textrm{for}\ \forall\ {s}'\in\mathbb Z_p,\ m',s\in\mathbb Z_q.
\end{eqnarray}
In addition, equation (\ref{s1'=0}) implies
\begin{eqnarray}
\mu_{ s0}-\mu_{s{s}'}=\mu_{00}-\mu_{0{s}'},\ \textrm{for}\ \forall\ {s}'\in\mathbb Z_p,\  s\in\mathbb Z_q. \nonumber
\end{eqnarray}
Therefore,
\begin{eqnarray}
\label{mu-s-s'}
\mu_{ s{s}'}=\mu_{s0}-\mu_{00}+\mu_{0{s}'},\ \textrm{for}\ \forall\ {s}'\in\mathbb Z_p,\ s\in\mathbb Z_q.
\end{eqnarray}
Substituting equation (\ref{mu-s-s'}) into equation (\ref{eta=mu+eta-m'0}), we have
\begin{eqnarray}
\label{eta=mu00-mu+eta}
\eta_{m'{s}'} =\mu_{00}-\mu_{0{s}'}+\eta_{m'0},\ \textrm{for}\ \forall\ {s}'\in\mathbb Z_p,\ m'\in\mathbb Z_q.
\end{eqnarray}
Equations (\ref{mu-s-s'}) and (\ref{eta=mu00-mu+eta}) mean that variables $\mu_{s{s}'}$ and $\eta_{m'{s}'}$ can be expressed by independent variables  $\mu_{s0}$, $\mu_{0{s}'}$ and $\eta_{m'0}$, for all $s\in\mathbb Z_q$,  ${s}'\in\mathbb Z_p$,  $m'\in\mathbb Z_q$.

Similarly, equation (\ref{eq-mm'}) implies
\begin{equation}
\lambda_{mm'}+\mu_{s{s}'}+\gamma_{ms}+\eta_{m'{s}'} =\lambda_{mm'}+\mu_{0{s}'}+\gamma_{m0}+\eta_{m'{s}'},\nonumber
\end{equation}
for all $m, s'\in\mathbb Z_p$, $m',s\in\mathbb Z_q$. Then
\begin{eqnarray}
\label{s2=0}
\mu_{0{s}'}-\mu_{s{s}'} =\gamma_{ms}-\gamma_{m0},\ \textrm{for}\ \forall\ m, s'\in\mathbb Z_p,\  s\in\mathbb Z_q.\nonumber
\end{eqnarray}
Thus
\begin{eqnarray}
\label{gam=mu-mu+gam-m0}
\gamma_{m s}=\mu_{0{s}'}-\mu_{s{s}'}+\gamma_{m0},\ \textrm{for}\ \forall\ m, s'\in\mathbb Z_p,\  s\in\mathbb Z_q.
\end{eqnarray}
Substituting equation (\ref{mu-s-s'}) into equation (\ref{gam=mu-mu+gam-m0}), we have
\begin{eqnarray}
\label{gam=00-s0-m0}
\gamma_{ms}=\gamma_{m0}-\mu_{s0}+\mu_{00},\ \textrm{for}\ \forall\ m\in\mathbb Z_p,\ s\in\mathbb Z_q.
\end{eqnarray}
By equation (\ref{eq-mm'}), one can obtain
\begin{eqnarray}
\label{lam-mm'=mu+gam+eta}
\lambda_{mm'}=-(\mu_{ss'}+\gamma_{ms}+\eta_{m's'}),\ \textrm{for}\ \forall\ m,s'\in\mathbb Z_p, m',s\in\mathbb Z_q.
\end{eqnarray}
Substitute equation (\ref{mu-s-s'}), (\ref{eta=mu00-mu+eta}) and (\ref{gam=00-s0-m0}) into equation (\ref{lam-mm'=mu+gam+eta}), then
\begin{equation}
\label{lam-mm'}
\lambda_{mm'}=-(\mu_{00}+\eta_{m'0}+\gamma_{m0}),\ \textrm{for}\ \forall\ m\in\mathbb Z_p, \ m'\in\mathbb Z_q.
\end{equation}
Equations (\ref{mu-s-s'}), (\ref{eta=mu00-mu+eta}), (\ref{gam=00-s0-m0}) and (\ref{lam-mm'}) mean that all variables $\lambda_{mm'}$, $\mu_{ss'}$, $\eta_{m's'}$ and $\gamma_{ms}$ can be expressed by  variables $\mu_{s0}$, $\mu_{0{s}'}$, $\eta_{m'0}$ and $\gamma_{m0}$ for all $s\in\mathbb Z_q$, ${s}'\in\mathbb Z_p$, $m'\in\mathbb Z_q$ and $m\in\mathbb Z_p$. For $\mu_{s0}$ and $\mu_{0{s}'}$, the same variable $\mu_{00}$ can be obtained when $s=0$ and $s'=0$. Therefore, there are $q+p+q+p-1=2(p+q)-1$ independent variables and $\dim({\rm Ker}(\Gamma))=2(p+q)-1$. Then $\dim{\rm \rm span}_{\mathbb{R}}(\mathcal{A}\cup \mathcal{B}\cup \mathcal{C}\cup \mathcal{D}) = 4pq-(2(p+q)-1)=(2p-1)(2q-1)$.

Next, we calculate the dimension of $\rm{KD}_{\mathcal{A},\mathcal{B}}^r$. Similar to the proof of theorem \ref{theorem},  we will apply lemma \ref{lemma-4} to discuss the dimension of $\rm{KD}_{\mathcal{A},\mathcal{B}}^r$.
Suppose
 $F \in \rm{KD}_{\mathcal{A},\mathcal{B}}^r$, for $d=pq$ $(2<p<q)$,
 all off-diagonal entries are complex in $F$. Particularly, for $d=2p$,  equation (\ref{F-k-(d-k)}), i.e., $F_{i(i+k)}=\overline {F_{(i+k)i}}=\overline {F_{(i-(d-k))i}} = \overline {F_{i(i+(d-k))}}
$, implies  $F_{i(i+p)}=\overline {F_{i(i+p)}}$ for all $i\in\mathbb Z_{2p}$. It implies that $F_{i(i+p)}$ is also real except for the diagonal entries in the matrix $F$.
 So we consider  two different cases $d=pq$ $(2 < p < q)$ and $d=2p$ $(2 < p)$ separately.

Case 1. $d=pq$ $(2 < p < q)$. For off-diagonal entries in $F$, we consider three cases gcd$(k,d)=1$, gcd$(k,d)=p$ and gcd$(k,d)=q$.

(i) gcd$(k,d)=1$. $k$ cannot take $0, p, 2p, \ldots$, $(q-1)p$, $q, 2q, \ldots, (p-1)q$.   Hence  $k$ can take $pq-p-q+1$ values. The length of equation (\ref{F-k}) is $pq$. By equation (\ref{F-k-(d-k)}), we have $F_{i(i+k)}= \overline {F_{i(i+(pq-k))}}$. Thus $F_{i(i+k)}$ and $F_{i(i+(pq-k))}$ are in a common category for all $i\in\mathbb Z_{pq}$. Therefore, $F_{i(i+k)}$ for all $i\in\mathbb Z_{pq}$ can be classified into $\frac{pq-p-q+1}{2}$ categories.

(ii) gcd$(k,d)=p$, i.e., $k=p, 2p, \ldots$, $(q-1)p$.
If $k=p$, then
\begin{equation}
\label{pq-F-gcd=p}
F_{(i-p)i}=F_{i(i+p)}=F_{(i+p)(i+2p)}=\cdots=F_{(i+(q-1)p)i}
\end{equation}
by equation (\ref{F-kt}). The length of equation (\ref{pq-F-gcd=p}) is $q$. Equation (\ref{pq-F-gcd=p}) means that $F_{i(i+p)}=F_{i'(i'+p)}$ when $i\equiv i' \mod p$, where $i'\in\mathbb Z_{p}$.
Therefore, $F_{i(i+p)}$ for all $i\in\mathbb Z_{pq}$ can be classified into $p$ categories and each category has $q$ entries. By equation (\ref{F-k-(d-k)}), we have $F_{i(i+p)}= \overline {F_{i(i+(pq-p))}}$. Thus $F_{i(i+p)}$ and $F_{i(i+(pq-p))}$ for all $i\in\mathbb Z_{pq}$ are in a common category. Thus, there are still $p$ categories and each category has $2q$ entries. Similarly, $F_{i(i+k)}$ and $F_{i(i+(pq-k))}$ for all $i\in\mathbb Z_{pq}$ can be classified into $p$ categories  and each of which has $2q$ entries. Therefore, $F_{i(i+k)}$ for all $i\in\mathbb Z_{pq}$ can be classified into $\frac{(q-1)p}{2}$ categories and each category has $2q$ entries.

(iii) gcd$(k,d)=q$, i.e., $k=q, 2q, \ldots, (p-1)q$. Repeat the process of gcd$(k,d)=p$, it is easy to verify that $F_{i(i+k)}$ for all $i\in\mathbb Z_{pq}$ can be classified into $\frac{(p-1)q}{2}$ categories and each category has $2p$ entries.

Off-diagonal entries in $F$  can be classified into $\frac{pq-p-q+1}{2}+\frac{(q-1)p}{2}+\frac{(p-1)q}{2}=\frac{3pq-2(p+q)+1}{2}$ categories. Every category is determined by a nonreal value. Together with $pq$ independent real diagonal entries, there are $2\times\frac{3pq-2(p+q)+1}{2}+pq= (2p-1)(2q-1)$ independent real parameters. Thus $\dim$ KD$_{\mathcal{A},\mathcal{B}}^r= (2p-1)(2q-1)$.

Case 2. $d=2p$ $(2 < p)$. For off-diagonal entries in $F$, we consider three cases gcd$(k,d)=1$, gcd$(k,d)=2$ and gcd$(k,d)=p$.

(i) gcd$(k,d)=1$. Employing a similar discussion to that in (i) of Case 1, we can obtain that $F_{i(i+k)}$ for all $i\in\mathbb Z_{2p}$ can be classified into $\frac{p-1}{2}$ categories and each category has $4p$ entries. Every category is determined by a nonreal value. Therefore, there are $(p-1)$ independent real parameters.

(ii) gcd$(k,d)=2$. Applying a similar discussion to that in (ii) of Case 1, we have that $F_{i(i+k)}$ for all $i\in\mathbb Z_{2p}$ can be classified into $p-1$ categories and each category has $2p$ entries. Every category is determined by a nonreal value. Therefore, there are $2(p-1)$ real parameters.

(iii) gcd$(k,d)=p$, i.e., $k=p$.
From equation (\ref{F-kt}) we have
\begin{equation}
\label{gcd=p-d=2p}
F_{i(i+p)}=F_{(i+p)i}.
\end{equation}
The length of equation (\ref{gcd=p-d=2p}) is $2$. Equation (\ref{gcd=p-d=2p}) means $F_{i(i+p)}=F_{i'(i'+p)}$ when $i\equiv i' \mod p$, where $i'\in\mathbb Z_p$. Therefore, $F_{i(i+p)}$ for all $i\in\mathbb Z_{2p}$ can be classified into $p$ categories and each category has two entries. By equation (\ref{F-k-(d-k)}), $F_{i(i+p)}=\overline {F_{i(i+p)}}$ for all $i\in\mathbb Z_{2p}$. It means that $F_{i(i+p)}$ is real. Therefore, there are $p$ real parameters. Together with $2p$ independent real diagonal entries, there are $(p-1)+2(p-1)+p+2p=3(2p-1)$ independent real parameters in $F$. Thus $\dim \rm{KD}_{\mathcal{A},\mathcal{B}}^r= 3(2p-1)$.
Since $\rm span_{\mathbb{R}}(\mathcal{A}\cup \mathcal{B}\cup \mathcal{C}\cup \mathcal{D})\subseteq \rm{KD}_{\mathcal{A},\mathcal{B}}^r$,  it follows $\rm span_{\mathbb{R}}(\mathcal{A}\cup \mathcal{B}\cup \mathcal{C}\cup \mathcal{D})= \rm{KD}_{\mathcal{A},\mathcal{B}}^r$  for $d=pq$. This completes the proof.
\end{proof}

\section{The proof of Theorem \ref{theorem-pq}}
\label{Proof of theorem-pq}
\begin{proof}
Suppose $\mathcal{X}=\mathcal{B}$, $\mathcal{Y}=\mathcal{C}$, $\mathcal{Z}=\mathcal{D}$. Other cases can be similarly proved.
Let us first calculate the dimension of $\textrm{span}_{\mathbb{R}}(\mathcal{B}\cup\mathcal{C}\cup\mathcal{D})$.   Define a linear map $\Gamma:\mathbb{R}^{3pq}\rightarrow \rm{span}_{\mathbb{R}}(\mathcal{B}\cup\mathcal{C}\cup\mathcal{D})$, for each $((\mu_j),(\gamma_{ms}),(\eta_{m's'}))\in\mathbb{R}^{3pq}$,
\begin{equation}
\Gamma((\mu_j),(\gamma_{ms}),(\eta_{m's'}))=\sum_{j=0}^{pq-1}\mu_jb_j+\sum_{{m\in\mathbb Z_p}\atop{s\in\mathbb Z_q}}\gamma_{ms}\psi_{ms}+\sum_{{m'\in\mathbb Z_q}\atop{s'\in\mathbb Z_p}}\eta_{m's'}\varphi_{m's'}.\nonumber
\end{equation}
From the rank theorem, we have $\dim{\rm{span}}_{\mathbb{R}}(\mathcal{B}\cup\mathcal{C}\cup\mathcal{D})= 3pq-\dim(\rm{Ker} (\Gamma))$.

Suppose
\begin{equation}
\Upsilon=\sum_{j=0}^{pq-1}\mu_jb_j+\sum_{{m\in\mathbb Z_p}\atop{s\in\mathbb Z_q}}\gamma_{ms}\psi_{ms}+\sum_{{m'\in\mathbb Z_q}\atop{s'\in\mathbb Z_p}}\eta_{m's'}\varphi_{m's'}=0, \nonumber
\end{equation}
where $\mu_j,\gamma_{ms},\eta_{m's'}\in \mathbb{R}$. We will calculate $\dim({\rm Ker}(\Gamma))$ below.
Since $\Upsilon=0$, one has
\begin{equation}
\langle a_i|\Upsilon|b_j\rangle=\langle a_i|b_j\rangle(\mu_j+\gamma_{ms}+\eta_{m's'}) = 0,\ \forall\ i,j\in\mathbb Z_{pq},\nonumber
\end{equation}
where
\begin{eqnarray}
\left\{
\begin{aligned}
i \equiv m \mod p \\
i \equiv m' \mod q
\end{aligned}
\right.,
\ \ \
\left\{
\begin{aligned}
j \equiv s \mod q \\
j \equiv s' \mod p
\end{aligned}
\right..\nonumber
\end{eqnarray}
According to Chinese remainder theorem,  $j$ is uniquely determined in $\mathbb Z_{pq}$ by $s$ and $s'$. Since $\langle a_i|b_j\rangle=\frac{1}{\sqrt{d}}\omega_d^{ij} \neq 0$, one has $\mu_j+\gamma_{ms}+\eta_{m's'} = 0$, i.e.,
\begin{equation}
\label{lam+mu+ga=0}
\mu_{ss'}+\gamma_{ms}+\eta_{m's'} = 0,  \ \textrm{for}\  \forall\ m,s'\in\mathbb Z_p,\ m',s\in\mathbb Z_q.
\end{equation}
It follows $\mu_{ss'}+\gamma_{ms}+\eta_{m's'} =\mu_{ss'}+\gamma_{0s}+\eta_{m's'} $. Thus $\gamma_{ms}=\gamma_{0s}$, for $\forall\ m\in\mathbb Z_p,\ s\in\mathbb Z_q$.
Thus equation (\ref{lam+mu+ga=0}) can be written as
\begin{equation}
\label{ga-0s}
\mu_{ss'}+\gamma_{0s}+\eta_{m's'} = 0,  \ \textrm{for}\  \forall\ s'\in\mathbb Z_p,\ m',s\in\mathbb Z_q.
\end{equation}
It follows $\mu_{ss'}+\gamma_{0s}+\eta_{m's'}=\mu_{0s'}+\gamma_{00}+\eta_{m's'}$. Hence
\begin{equation}
\label{mu-ss'}
\mu_{ss'}=\mu_{0s'}+\gamma_{00}-\gamma_{0s},  \ \textrm{for}\  \forall\ s'\in\mathbb Z_p,\ s\in\mathbb Z_q.
\end{equation}
Substituting equation (\ref{mu-ss'}) into equation (\ref{ga-0s}), then
\begin{equation}
\label{eta}
\eta_{m's'}=-(\mu_{0s'}+\gamma_{00}),  \ \textrm{for}\  \forall\ s'\in\mathbb Z_p,\ m'\in\mathbb Z_q.
\end{equation}

As can be seen from equation (\ref{mu-ss'}) and  equation (\ref{eta}), all variables $\mu_{ss'}$, $\eta_{m's'}$ and $\gamma_{ms}$ can be expressed by  variables  $\mu_{0{s}'}$ and $\gamma_{0s}$ for all  ${s}'\in\mathbb Z_p$ and $s\in\mathbb Z_q$. one can obtain $\dim({\rm{Ker}} (\Gamma))=p+q$, then $\dim\rm{span}_{\mathbb{R}}(\mathcal{B}\cup \mathcal{C}\cup \mathcal{D}) =$ $3pq - p-q$.

Next, we prove the second statement. By equation (\ref{KD-probali}), it is easy to verify $\rm{ConvHull}(\mathcal{B}\cup \mathcal{C}\cup \mathcal{D})\subseteq \rm{KD}_{\mathcal{A},\mathcal{B}}^+$. Together with the fact that $\rm{ConvHull}(\mathcal{B}\cup \mathcal{C}\cup \mathcal{D})\subseteq \rm{span}_{\mathbb{R}}(\mathcal{B}\cup \mathcal{C}\cup \mathcal{D})$, it means $\rm{ConvHull}( \mathcal{B}\cup \mathcal{C}\cup \mathcal{D})\subseteq \rm{KD}_{\mathcal{A},\mathcal{B}}^+\cap \rm{span}_{\mathbb{R}}(\mathcal{B}\cup \mathcal{C}\cup \mathcal{D})$. Now we only need to prove $\rm{KD}_{\mathcal{A},\mathcal{B}}^+\cap \rm{span}_{\mathbb{R}}( \mathcal{B}\cup \mathcal{C}\cup \mathcal{D})\subseteq \rm{ConvHull}( \mathcal{B}\cup \mathcal{C}\cup \mathcal{D})$.

Let
\begin{equation}
\rho\in\rm{KD}_{\mathcal{A},\mathcal{B}}^+\cap \rm{span}_{\mathbb{R}}( \mathcal{B}\cup \mathcal{C}\cup \mathcal{D}).\nonumber
\end{equation}
There exists a vector $( (\mu_j), (\gamma_{ms}), (\eta_{m's'}))\in\mathbb{R}^{3pq}$ so that
\begin{equation}
\label{rho-span3}
\rho=\sum_{j=0}^{d-1}\mu_jb_j+\sum_{m,s}\gamma_{ms}\psi_{ms}+\sum_{m',s'}\eta_{m's'}\varphi_{m's'},\forall\ m,s'\in \mathbb Z_p,\forall\ m',s\in \mathbb Z_q.
\end{equation}
Since $\rho\in\rm{KD}_{\mathcal{A},\mathcal{B}}^+$, it implies $Q_{ij}(\rho)=|\langle a_i|b_j\rangle|^2(\mu_j+\gamma_{ms}+\eta_{m's'}) \geqslant 0,\forall\ i,j\in\mathbb Z_{pq}$,
 where
\begin{eqnarray}
\left\{
\begin{aligned}
i \equiv m \mod p \\
i \equiv m' \mod q
\end{aligned}
\right.,
\ \ \
\left\{
\begin{aligned}
j \equiv s \mod q \\
j \equiv s' \mod p
\end{aligned}
\right..\nonumber
\end{eqnarray}
According to Chinese remainder theorem,  $j$ is uniquely determined in $\mathbb Z_{pq}$ by $s$ and $s'$.
Equation (\ref{rho-span3}) can be written as
\begin{equation}
\label{j-ss'}
\rho=\sum_{s,s'}\mu_{ss'}b_{ss'}+\sum_{m,s}\gamma_{ms}\psi_{ms}+\sum_{m',s'}\eta_{m's'}\varphi_{m's'},\forall\ m,s'\in \mathbb Z_p,\forall\ m',s\in \mathbb Z_q.
\end{equation}
Since $|\langle a_i|b_j\rangle|^2\geq0$, one has $\mu_j+\gamma_{ms}+\eta_{m's'} \geq0$, i.e.,
\begin{equation}
\label{mu+gam+eta>=0}
\mu_{ss'}+\gamma_{ms}+\eta_{m's'} \geq 0,  \ \textrm{for}\  \forall\ m,s'\in\mathbb Z_p,\forall\ m',s\in\mathbb Z_q.
\end{equation}

In order to see that the coefficients of $\gamma_{ms}$ and $\eta_{m's'}$ are non-negative, without loss of generality, suppose that ${\rm min}_{m\in\mathbb Z_p}\{\gamma_{ms}\}=\gamma_{0s}$ for $s\in\mathbb Z_q$ and ${\rm min}_{m'\in\mathbb Z_q}\{\eta_{m's'}\}=\eta_{0s'}$ for $s'\in\mathbb Z_p$. Then equation (\ref{j-ss'}) can be written as
\begin{eqnarray}
\label{eq15}
\rho &=&\sum_{s,s'}\mu_{ss'}b_{ss'}+\sum_{m,s}(\gamma_{ms}-\gamma_{0s})\psi_{ms}+\sum_{m',s'}(\eta_{m's'}-\eta_{0s'})\varphi_{m's'} \nonumber\\
&\;&+\sum_{m,s}\gamma_{0s}\psi_{ms}+\sum_{m',s'}\eta_{0s'}\varphi_{m's'}
\end{eqnarray}
Notice that $\gamma_{ms}-\gamma_{0s}\geq0$ and $\eta_{m's'}-\eta_{0s'}\geq0$.
Now let us consider the relation among the projectors  $b_{ss'}$, $\psi_{ms}$ and $\varphi_{m's'}$. By
equations (\ref{pq-psi})and (\ref{pq-varphi}), we have
\begin{eqnarray}
\psi_{ms}&=&\frac{1}{q}(\sum_{k=0}^{q-1}a_{kp+m}+\sum_{k_1\neq k_2}\omega_q^{s(k_1-k_2)} |a_{k_{1}p+m}\rangle\langle a_{k_{2}p+m}|) \nonumber\\
&=&\frac{1}{p}(\sum_{l=0}^{p-1}b_{lq+s}+\sum_{l_1\neq l_2}\omega_p^{-m(l_1-l_2)}|b_{l_{1}q+s}\rangle\langle b_{l_{2}q+s}|),\nonumber
\end{eqnarray}
\begin{eqnarray}
\varphi_{m's'}&=&\frac{1}{p}(\sum_{k=0}^{p-1}a_{kq+m'}+\sum_{k_1\neq k_2}\omega_p^{s'(k_1-k_2)} |a_{k_{1}q+m'}\rangle\langle a_{k_{2}q+m'}|) \nonumber\\
&=&\frac{1}{q}(\sum_{l=0}^{q-1}b_{lp+s'}+\sum_{l_1\neq l_2}\omega_q^{-m'(l_1-l_2)}|b_{l_{1}p+s'}\rangle\langle b_{l_{2}p+s'}|).\nonumber
\end{eqnarray}
Hence,
\begin{equation}
\forall s\in\mathbb Z_q, \sum_{m=0}^{p-1}\psi_{ms}=\sum_{l=0}^{p-1}b_{lq+s}=\sum_{s'}b_{ss'}, \nonumber
\end{equation}
\begin{equation}
\forall s'\in\mathbb Z_p, \sum_{m'=0}^{q-1}\varphi_{m's'}=\sum_{l=0}^{q-1}b_{lp+s'}=\sum_{s}b_{ss'}. \nonumber
\end{equation}
Then
\begin{equation}
\label{psi=b}
\sum_{m,s}\gamma_{0s}\psi_{ms}=\sum_{s,s'}\gamma_{0s}b_{ss'},
\end{equation}
\begin{equation}
\label{varphi=b}
\sum_{m',s'}\eta_{0s'}\varphi_{m's'}=\sum_{s,s'}\eta_{0s'}b_{ss'}
\end{equation}
Substituting equations (\ref{psi=b}) and  (\ref{varphi=b}) into equation (\ref{eq15}), one can obtain
\begin{equation}
 \label{eq17}
\rho
=\sum_{s,s'}(\mu_{ss'}+\gamma_{0s}+\eta_{0s'})b_{ss'}+\sum_{m,s}(\gamma_{ms}-\gamma_{0s})\psi_{ms}+\sum_{m',s'}(\eta_{m's'}-\eta_{0s'})\varphi_{m's'}.
\end{equation}
Note that $\mu_{ss'}+\gamma_{0s}+\eta_{0s'}\geq0$ from equation (\ref{mu+gam+eta>=0}). Then all coefficients of $\rho$ in equation (\ref{eq17}) are
non-negative. Hence $\rho\in$ span$_{\mathbb{R^{+}}}( \mathcal{B}\cup \mathcal{C}\cup \mathcal{D})$. Thus $\rho\in$ ConvHull$(
\mathcal{B}\cup \mathcal{C}\cup \mathcal{D})$ since ${\rm Tr}\rho=1$. This completes the proof
\end{proof}


\section*{References}

\begin{thebibliography}{99}
\bibitem{negativity.1}
Lostaglio M, Belenchia A, Levy A, Hern\'{a}ndez-G\'{o}mez S, Fabbri N and Gherardini S 2023 Kirkwood-Dirac quasiprobability approach to the statistics of incompatible observables  {\it Quantum} {\bf  7} 1128
\bibitem{negativity.2}
Wagner R, Schwartzman-Nowik Z, Paiva I L, Te'eni A, Ruiz-Molero A, Barbosa R S, Cohen E and Galv\~{a}o E F 2024 Quantum circuits for measuring weak values, Kirkwood-Dirac quasiprobability distributions, and state spectra {\it Quantum Sci. Technol.}  {\bf 9} 015030
\bibitem{negativity.3}
Gherardini S and De Chiara G 2024 Quasiprobabilities in quantum thermodynamics and many-body systems: A tutorial (arXiv:2403.17138)
\bibitem{Arvidsson.2024}
Arvidsson-Shukur D R, Braasch Jr W F, De Bi\`{e}vre S, Dressel J, Jordan A N, Langrenez C, Lostaglio M, Lundeen J S and  Halpern N Y 2024 Properties and Applications of the Kirkwood-Dirac Distribution (arXiv: 2403.18899)
\bibitem{Halpern.2018}
Halpern N Y, Swingle B and Dressel J 2018 Quasiprobability behind the out-of-time-ordered correlator {\it Phys. Rev.} A {\bf 97} 042105
\bibitem{Budiyono.2023}
Budiyono A, Sumbowo J F, Agusta M K and Nurhandoko B E 2023 Characterizing quantum coherence based on the negativity and nonreality of the Kirkwood-Dirac
quasiprobability (arXiv:2309.09162)
\bibitem{entanglement}
Horodecki R, Horodecki P, Horodecki M and Horodecki K 2009 Quantum entanglement {\it Rev. Mod. Phys.} {\bf 81} 865
\bibitem{discord}
Ollivier H and Zurek W H 2001 Quantum discord: a measure of the quantumness of correlations {\it Phys. Rev. Lett.} {\bf 88} 017901
\bibitem{coherence}
Baumgratz T, Cramer M and Plenio M B 2014 Quantifying coherence {\it Phys. Rev. Lett.} {\bf 113} 140401
\bibitem{coherence.2017}
Streltsov A, Adesso G and Plenio M B 2017 Colloquium: Quantum coherence as a resource {\it Rev. Mod. Phys.} {\bf 89} 041003
\bibitem{nonlocality}
Brunner N, Cavalcanti D, Pironio S, Scarani V and Wehner S 2014 Bell nonlocality {\it Rev. Mod. Phys.} {\bf 86} 419
\bibitem{uncertainty}
Robertson H P 1929 The uncertainty principle {\it Phys. Rev.} {\bf 34} 163
\bibitem{Kirkwood.1933}
Kirkwood J G 1933 Quantum statistics of almost classical assemblies {\it Phys. Rev.} {\bf 44} 31
\bibitem{Dirac.1945}
Dirac P A 1945 On the analogy between classical and quantum mechanics {\it Rev. Mod. Phys.} {\bf 17} 195
\bibitem{Kolmogorov.1951}
Brookes B C and Kolmogorov A N 1951 Foundations of the theory of probability {\it The Math. Gaz.} {\bf 35} 292
\bibitem{Lundeen.2012}
Lundeen J S and Bamber C 2012 Procedure for direct measurement of general quantum states using weak measurement {\it Phys. Rev. Lett.} {\bf 108} 070402
\bibitem{Bamber.2014}
Bamber C and Lundeen J S 2014 Observing Diracs classical phase space analog to the quantum state {\it Phys. Rev. Lett.} {\bf 112} 070405
\bibitem{Thekkadath.2016}
Thekkadath G S, Giner L, Chalich Y, Horton M J, Banker J and Lundeen J S 2016 Direct measurement of the density matrix of a quantum system {\it Phys. Rev. Lett.} {\bf 117} 120401
\bibitem{Pusey.2014}
Pusey M F 2014 Anomalous weak values are proofs of contextuality {\it Phys. Rev. Lett.} {\bf 113} 200401
\bibitem{Dressel.2015}
Dressel J 2015 Weak values as interference phenomena {\it Phys. Rev.} A {\bf 91} 032116
\bibitem{Lupu.2022}
Lupu-Gladstein N, Yilmaz Y B, Arvidsson-Shukur D R, Brodutch A, Pang A O, Steinberg A M and Halpern N Y 2022 Negative quasiprobabilities enhance phase estimation in quantum optics experiment {\it Phys. Rev. Lett.} {\bf 128} 220504
\bibitem{Arvidsson.2020}
Arvidsson-Shukur D R, Yunger Halpern N, Lepage H V, Lasek A A, Barnes C H and Lloyd S 2020 Quantum advantage in postselected metrology {\it Nat. Commun.}
  {\bf 11} 3775
\bibitem{Jenne.2022}
Jenne J H and Arvidsson-Shukur D R 2022 Unbounded and lossless compression of multiparameter quantum information {\it Phys. Rev.} A {\bf 106} 042404
\bibitem{Levy.2020}
Levy A and Lostaglio M 2020 Quasiprobability distribution for heat uctuations in the quantum regime {\it PRX Quantum} {\bf 1} 010309
\bibitem{Lostaglio.2020}
Lostaglio M 2020 Certifying quantum signatures in thermodynamics and metrology via contextuality of quantum linear response {\it Phys. Rev. Lett.} {\bf 125} 230603
\bibitem{Spekkens.2008}
Spekkens R W 2008 Negativity and contextuality are equivalent notions of nonclassicality {\it Phys. Rev. Lett.} {\bf 101} 020401
\bibitem{Lostaglio.2018}
Lostaglio M 2018 Quantum fluctuation theorems, contextuality, and work quasiprobabilities {\it Phys. Rev. Lett.} {\bf 120} 040602
\bibitem{Schmid.KD.representations}
Schmid D, Baldij\~{a}o R D, Y\={\i}ng Y, Wagner R and Selby J H 2024 Kirkwood-Dirac representations beyond quantum states (and their relation to noncontextuality) (arXiv:2405.04573)
\bibitem{ArvidssonJPA.2021}
Arvidsson-Shukur D R, Drori J C and Halpern N Y 2021 Conditions tighter than noncommutation needed for nonclassicality {\it J. Phys. A: Math. Theor.} {\bf 54} 284001
\bibitem{BievrePRL.2021}
De Bi\`{e}vre S 2021 Complete incompatibility, support uncertainty, and Kirkwood-Dirac nonclassicality {\it Phys. Rev. Lett.} {\bf 127} 190404
\bibitem{BievreJMP.2023}
De Bi\`{e}vre S 2023 Relating incompatibility, noncommutativity, uncertainty, and Kirkwood-Dirac nonclassicality {\it J. Math. Phys.} {\bf 64} 022202
\bibitem{Xu.PLA2024}
Xu J W 2024 Kirkwood-Dirac classical pure states {\it Phys. Lett.} A {\bf 510} 129529
\bibitem{Yang.2023}
Yang Y H, Zhang B B, Wang X L, Geng S J and Chen P Y 2023 Characterizing Kirkwood-Dirac nonclassicality and uncertainty diagram based on discrete Fourier transform {\it Entropy} {\bf 25} 1075
\bibitem{XuPRA.2022}
Xu J W 2022 Classification of incompatibility for two orthonormal bases {\it Phys. Rev.} A {\bf 106} 022217
\bibitem{Fiorentino.2022}
Fiorentino V and Weigert S 2022 Uncertainty relations for the support of quantum states {\it J. Phys. A: Math. Theor.} {\bf 55} 495305
\bibitem{Langrenez.2023}
Langrenez C, Arvidsson-Shukur D R and De Bi\`{e}vre S 2024 Characterizing the geometry of the Kirkwood-Dirac-positive states {\it J. Math. Phys.} {\bf 65} 072201
\bibitem{Bargmann.2001}
Mukunda N, Chaturvedi S and Simon R 2001 Bargmann invariants and off-diagonal geometric phases for multilevel quantum systems: A unitary-group approach {\it Phys. Rev.} A {\bf 65} 012102
\bibitem{Bargmann.2003}
Mukunda N, Ercolessi E, Marmo G, Morandi G and Simon R. Bargmann invariants, null phase curves, and a theory of the geometric phase {\it Phys. Rev.} A {\bf 67} 042114
\bibitem{Bargmann.2020}
Akhilesh K S, Arvind A, Chaturvedi S, Mallesh K S and Mukunda N 2020 Geometric phases for finite-dimensional systems---The roles of Bargmann invariants, null phase curves, and the Schwinger-Majorana SU (2) framework {\it J. Math. Phys.} {\bf 61} 072103
\bibitem{Bargmann.2024}
Oszmaniec M, Brod D J and Galv\~{a}o E F 2024 Measuring relational information between quantum states, and applications  {\it New J. Phys.} {\bf 26} 013053
\bibitem{Bargmann.arXiv2024}
Fernandes C, Wagner R, Novo L and Galv\~{a}o E F 2024 Unitary-invariant witnesses of quantum imaginarity (arXiv:2403.15066)
\bibitem{Hiriart.2004}
Hiriart-Urruty J B and Lemar\'{e}chal C 2004 Fundamentals of convex analysis {\it Spr. Sci. Bus. Med.}
\bibitem{Xu.2023}
Xu J W 2022 Kirkwood-Dirac classical pure states (arXiv:2210.02876)
\bibitem{He-J.2024}
He J Y and Fu S S 2024 Nonclassicality of the Kirkwood-Dirac quasiprobability distribution via quantum modification terms {\it Phys. Rev.} A {\bf 109} 012215
\bibitem{Quasiprobability.heat.2020}
Levy A and Lostaglio M 2020  Quasiprobability distribution for heat fluctuations in the quantum regime  {\it  PRX Quantum} {\bf 1} 010309
\bibitem{heat.flow.2024}
Comar N E, Cius D, Santos L F, Wagner R and Amaral B 2024  Contextuality in anomalous heat flow (arXiv:2406.09715)
\end{thebibliography}
\end{document}